%% file: main.tex
\documentclass[conference]{IEEEtran}
\input{mydefs}

\begin{document}
\date{}
\title{\Large \bf \sysname: Knowledge Unlearning by Deviating Representation \\ for Large Language Models}
\author{
    Ce Fang\textsuperscript{1},
    Zhikun Zhang\textsuperscript{1},
    Min Chen\textsuperscript{2},
    Qing Liu\textsuperscript{1},
    Lu Zhou\textsuperscript{3},
    Zhe Liu\textsuperscript{1\S},
    Yunjun Gao\textsuperscript{1\S}
    \\[4pt]
    \textsuperscript{1}{Zhejiang University}
    \quad
    \textsuperscript{2}{Vrije Universiteit Amsterdam}
    \\
    \textsuperscript{3}{Nanjing University of Aeronautics and Astronautics}

    \\
    \{fangce, zhikun, qingliucs, zhe.liu, gaoyj\}@zju.edu.cn, m.chen2@vu.nl, lu.zhou@nuaa.edu.cn
}
\date{}

\pagestyle{plain}

\maketitle

\def\thefootnote{\#}\footnotetext{Corresponding authors.}
\renewcommand*{\thefootnote}{\arabic{footnote}}


\begin{abstract}
\textit{Large language models} (LLMs) acquire a large amount of knowledge via pre-training on vast corpora.
While this endows LLMs with strong capabilities in generation and reasoning, it amplifies risks associated with sensitive, copyrighted, or harmful content in training data.
LLM unlearning, which aims to remove specific knowledge encoded within models, is a promising technique to reduce these risks.
However, existing methods often force LLMs to generate random or incoherent answers due to their inability to alter the encoded knowledge precisely.
To achieve effective unlearning at the knowledge level, we propose \textit{\textbf{K}nowledge \textbf{U}nlearning by \textbf{D}eviating represent\textbf{A}tion} (\sysname), which synergizes parameter selection with knowledge removal based on the property of \textit{knowledge storage}.
We first utilize a component-level method and a sliding window to search specific layers for knowledge storage.
Then, we design a new unlearning objective that induces the model's representations to deviate from their original positions for knowledge removal, thus weakening the ability to associate with the target knowledge.
To resolve the optimization conflicts between forgetting and retention, we employ a relaxed null-space projection to mitigate the disruption to retaining knowledge.
Extensive experiments on representative benchmarks, WMDP and MUSE, demonstrate that \sysname outperforms most existing baselines by effectively balancing knowledge removal and model utility retention. 
Source code is available at \url{https://github.com/AceNagi/KUDA}.
\end{abstract}

\maketitle

\section{Introduction}
\label{sec:intro}

\textit{Large language models} (LLMs) have demonstrated great potential across various domains, such as scientific research, education, and economics.
These capabilities stem from the process of pre-training, during which LLMs learn and store extensive knowledge from diverse corpora~\cite{li2024wmdp}. 
However, the training corpora are primarily collected from the Internet, which may contain sensitive personal information, dangerous knowledge, or copyrighted content~\cite{DZCSJCCBZ25,trustllm,shi2024large,WZZTW25}. 
This raises societal safety concerns, as there is growing awareness that LLMs have a potential risk of being misused~\cite{du2025auditsok,faceaudit,scalable-carlini, YZDCSGHC26}.
Alignment training~\cite{ouyang2022training} represents the mainstream solution for safety risks. 
However, the phenomenon of shallow alignment~\cite{Qi2025shallow} renders such approaches inherently vulnerable to adversarial attacks. 
To achieve thorough safety, \textit{LLM unlearning} has emerged as a promising technique by removing the undesired knowledge stored in LLMs' parameters~\cite{liu2025rethinking,ren2025sokmachineunlearninglarge}.

\mypara{Existing Solutions}
The mainstream unlearning methods for LLMs are to delete or suppress target knowledge that is encoded in LLMs through updating parameters~\cite{ren2025sokmachineunlearninglarge, liu2025rethinking}.
The \textit{full parameter updates} methods update all parameters which leads LLMs to generate uncontrolled responses when queries are related to sensitive personal information or dangerous knowledge, which we refer to as \textit{target knowledge}~\cite{jang-etal-2023-knowledge, zhangnegative, fan2024simplicity}.
However, the indiscriminate updates risk severe degradations of model utility~\cite{ren2025sokmachineunlearninglarge}.
Furthermore, the response-oriented methods merely achieve response-level mitigation without inherently removing the target knowledge encoded in the model's parameters~\cite{hong-etal-2024-dissecting, ren2025sokmachineunlearninglarge, zhang2025focusonkey}.
This renders the removed knowledge easily restored with specific prompts~\cite{zhang2025focusonkey}.

In comparison, the \textit{partial parameter updates} approaches only update a portion of parameters to achieve more precise and targeted knowledge removal.
Nevertheless, they design the parameter selection and removal algorithm independently, failing to tradeoff the target knowledge removal and model utility retention~\cite{fan2025llmunlearningresilientrelearning,yang2025faithun}.
Concretely, these methods either employ well-designed localization approaches to identify knowledge-relevant parameters and directly apply loss functions from full parameter updates~\cite{WALGE, wu-etal-2023-depn}, or introduce fine-grained removal algorithms yet rely on empirical search for parameter selection~\cite{li2024wmdp}.
Therefore, there is a pressing need for a fine-grained, holistic approach that synergistically couples parameter selection and knowledge removal to balance the target knowledge removal and model utility retention.

\mypara{Our Proposal}
To this end, we propose \textit{\textbf{K}nowledge \textbf{U}nlearning by \textbf{D}eviating represent\textbf{A}tion} (\sysname), which synergizes parameter selection with knowledge removal based on the \textit{knowledge storage} property of LLMs.
Specifically, the \textit{feed-forward neural network} (FFN) components at specific layers are responsible for storing knowledge learned from training corpora~\cite{geva2021transformer}, whereas the \textit{multihead self-attention} (MHSA) components at some layers primarily facilitate semantic aggregation~\cite{geva2023dissecting}.
Therefore, by applying updates exclusively to FFNs specialized for knowledge storage, while preserving MHSA components intact, we can remove target knowledge encoded in parameters, with minimal degradation to the LLMs' general capabilities.
However, related works merely confine updates to FFNs' parameters~\cite{li2024wmdp, wu-etal-2023-depn}, lacking a synergy pipeline that integrates parameter selection and knowledge removal.

To bridge this gap, \sysname establishes a partial parameter update unlearning method operating in the representation space.
\sysname begins with a component-level identification method, consisting of causal tracing~\cite{meng2022locating} and the metric causal effect (CE), to pinpoint layers whose FFNs exhibit a significant contribution to knowledge storage. 
Subsequently, a sliding window search strategy is employed to determine unlearning layers, whose a subset of parameters is then designated for unlearning updates.
Building upon the selected layers, \sysname implements knowledge removal by intervening in the corresponding \textit{representations (\ie, the activation outputs of models' intermediate layers)}.
Since knowledge stored in FFNs is activated and incorporated into representations during generation~\cite{elazar-etal-2021-amnesic}, we design a knowledge removal loss function, which induces deviation to push the representation of target knowledge away from its original state.
This disrupts the model’s ability to utilize the knowledge for generation.
Simultaneously, a constraint loss is introduced to prevent excessive deviations.
Besides, we depart from the conventional use of a scaling factor to balance knowledge removal and utility retention, as it fails to reconcile the directional conflicts inherent between the two objectives.
Instead, we exploit the invariance property of the null-space, and introduce a relaxation threshold for knowledge entanglement.
By projecting the unlearning gradients into a relaxed null-space, \sysname eliminates gradient components detrimental to model utility, thereby removing target knowledge while preserving the model's capabilities.
Moreover, we propose a two-stage tuning strategy for the hyperparameter configuration by transforming the joint optimization into two decoupled linear searches, 
thus enabling \sysname to adapt more effectively to different scenario settings without the high cost of grid search.

\mypara{Evaluation}
We conduct experiments on two representative benchmarks to illustrate the superiority of \sysname.
Besides, we introduce a harmonic-mean-based metric, \textit{knowledge removal degree} (KRD), to capture the comprehensive forgetting effect from multiple perspectives, while maintaining compatibility with different benchmarks.
First, we showed that \sysname achieves a good trade-off between knowledge removal and utility retention. 
For instance, while achieving near-complete removal of target knowledge, \sysname incurs marginal utility degradation of only $3.1\%$ on MUSE and $1.7\%$ on WMDP.
In addition to unlearning performance, \sysname demonstrates robust model generalizability, seamlessly extending its unlearning capabilities to modern models like Llama-3.1 and Qwen-3.
From the perspective of gradient dynamics, we further reveal that \sysname balances removal and retention by maintaining near-orthogonal gradient directions between the two objectives, thereby preventing significant conflicts.
Moreover, we reveal the coupled impact of two influential hyperparameters, \ie, the inverse unlearning temperature and the null-space relaxation threshold, thus necessitating our decoupled search strategy.

\mypara{Contributions}
Our contributions are three-folds: 
\begin{itemize}
    
    \item Grounded in the LLMs' knowledge storage property, we propose a representation-based unlearning method that integrates \textit{unlearning layer selection}, \textit{representation deviation}, and \textit{relaxed null-space projection}.
    By applying stable updates to parameters correlated with knowledge storage, \sysname achieves targeted knowledge removal while retaining the LLMs' utility effectively.

    \item By identifying the stability boundary, we introduce a two-stage tuning strategy that simplifies the joint hyperparameter optimization into two decoupled linear searches. 
    This facilitates the practical deployment across different settings. 
    
    \item We conduct extensive experiments to demonstrate the unlearning performance of \sysname. 
    Furthermore, we provide mechanistic insights into gradient dynamics, revealing that \sysname mitigates unlearning conflicts by maintaining near-orthogonal gradients.
\end{itemize}


\section{Preliminary}
\label{sec:Preliminaries}

\begin{figure}[!t]
    \centering
    \begin{subfigure}{0.49\columnwidth}
    \includegraphics[width=\textwidth]{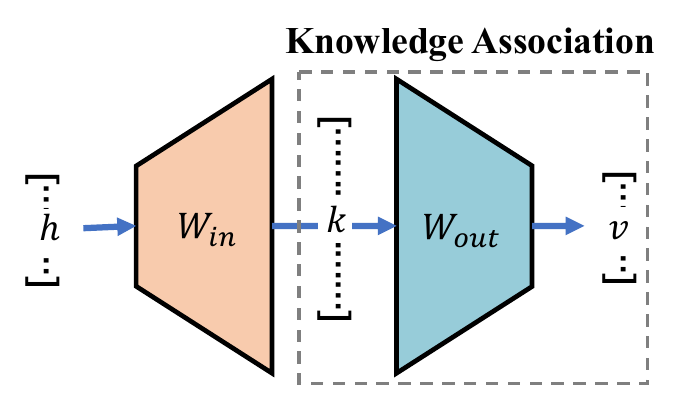}
    \caption{Classic Transformer}
    \label{fig:general ffn}
    \end{subfigure}
    \begin{subfigure}{0.49\columnwidth}
    \includegraphics[width=\textwidth]{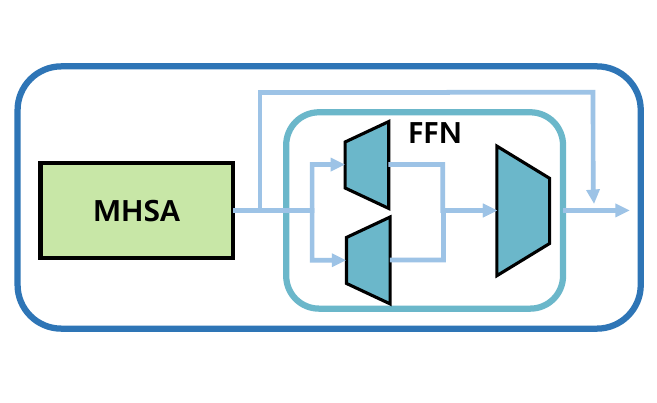}
    \caption{Llama series}
    \label{fig:llama ffn}
    \end{subfigure} \\ 
    \caption{FFN architectures in Transformer-based models. 
    }
    \label{fig:ffn structure}
\end{figure}

\subsection{Large Language Models}
\label{subsec: Large Language Model}
Transformer-based LLMs achieve context understanding and text generation with two crucial mechanisms: FFN-based associative memory and MHSA-driven contextual interaction~\cite{geva2021transformer}. 
Concretely, FFN acts as a 
knowledge storage repository~\cite{meng2022locating, geva2021transformer}, whereas MHSA facilitates the semantics aggregation~\cite{geva2023dissecting}.

\mypara{FFN for Knowledge Storage}
FFN can be modeled as \textit{a two-layer key–value knowledge storage}~\cite{meng2022locating}, as illustrated in \autoref{fig:general ffn}.
Given an input, the first-layer $W_{in}$ generates a \textit{key} based on the input; the second-layer $W_{out}$ maps the key to the stored knowledge related to that input and generates a corresponding \textit{value}.
This knowledge-retrieval process incorporates rich knowledge into hidden \textit{representations}, namely the activation outputs of models' intermediate layers, to influence the final text output, which can be formally formulated as:
\begin{equation}
\begin{aligned}
k^l_i &= \sigma(W^l_{in}(h^{l-1}_i)), \nonumber \\
v^l_i &= W^l_{out} \cdot k^l_i, \nonumber
\end{aligned}
\end{equation}
where $l$ and $i$ represent the layer index and the token position, and $\sigma$ denotes the activation function, such as ReLU. 
$W_{in}$ and $W_{out}$ are two-layer linear transformation matrices.

For more recent LLMs, such as the Llama series, that adopt advanced FFN variants, the mechanism of knowledge storage remains functionally equivalent.
As shown in \autoref{fig:llama ffn}, this FFN consists of three linear transformation matrices ($W_{gate}$, $W_{down}$, and $W_{up}$), and the process of generating value can be represented as:
\begin{equation}
v^l_i = W^l_{down} \cdot [\mathrm{Swish}(W^l_{gate}(h^{l-1}_i)) \otimes W^l_{up}(h^{l-1}_i)], \nonumber
\end{equation}
where $\otimes$ is the Hadamard Product and Swish is an activation function. 
Although differing structurally from the classical FFN, the combined operations of $W_{gate}$ and $W_{up}$ can still be regarded as the first layer to form the key, while $W_{down}$ serves as the second layer to transform key to value. 

\mypara{MHSA for Semantics Aggregation}
MHSA provides LLMs with powerful capabilities of semantic aggregation~\cite{geva2023dissecting}.
On one hand, it aggregates semantic information from preceding tokens to the current position, thus facilitating awareness of context.
On the other hand, it extracts task-relevant knowledge from the representations refined by FFNs' knowledge storage mechanism.
This synergy effect enables LLMs to understand the complex context and utilize the rich knowledge encoded in LLMs' parameters to generate high-quality text.

\subsection{Causal Tracing} 
\label{subsec: Causal Tracing}
\textit{Causal tracing}~\cite{meng2022locating} can quantify the importance of activations across layers and components for generating correct output.
The workflow mainly contains three stages:

\mypara{Clear Run} 
The objective is to collect the original states of generation.
We execute a forward pass with the given prompts to collect activations across all components, denoted as $H=\{h^{l}_{m,i}\mid m \in \mathrm{\{ MHSA,FFN,Layer \}}\}$, where $l$ is the layer index and $i$ is the token position.
Simultaneously, we record the probability $\mathbb{P}[t_{a_1}]$ of the model correctly predicting the first answer token $t_{a_1}$.

\mypara{Corrupted Run} 
The objective is to collect the corrupted states with the perturbed input embedding. 
We add perturbing noise to the embedding of input tokens to prevent the model from making a correct prediction.
Corresponding to ``Clear Run,'' collect the activations $H^*$ and the probability $\mathbb{P}^{*}$ of the correct first answer token in the current case. 

\mypara{Corrupted-with-Restoration Run} 
The objective is to restore each activation to quantify its importance for the correct prediction.
With the perturbed input embeddings, we individually replace the activations of each component across layers with the correct state from ``Clear Run,'' and record the resulting probability $\mathbb{P}^{*}_{m^l_i}[t_{a_1}]$, where $m^l_i$ is the restored component type $m$, layer $l$, and token $i$.

\subsection{Null-Space}
\label{subsec: null-space}
\textit{Null-space} is an effective technique for enhancing the stability and plasticity of deep learning networks in post-training scenarios~\cite{fangalphaedit, wang2021training}.
The core idea is to confine parameter updates within the null-space of the input feature space learned from previous tasks, thereby protecting prior knowledge from updated parameters.

For network $\mathcal{M}$ trained on prior tasks with data $x_{old}$, let $\Delta w^l$ denote the update gradients of a new task applied to parameters $W_{old}$ at layer $l$.
When $\Delta w^l$ lies within the null-space of the \textit{input features space} $X^l_{old}$ of previous tasks, there is $x^l_{old} \cdot \Delta w^l = 0$~\cite{wang2021training}.
The proof is given in Appendix \ref{sec: proof}.
This property prevents detrimental disruption to the representations at layer $l$ learned from previous tasks during continual learning:
\begin{equation}
x^l_{old} \cdot (W^l_{old} + \Delta w^l)  = x^l_{old} \cdot W^l_{old} + 0= y^l_{old} 
\nonumber
\end{equation}
Therefore, constraining gradients within such a null-space could preserve the integrity of the model's previous foundational capabilities:
\begin{equation}
\mathcal M (x_{old}, W_{old}) = \mathcal M'(x_{old}, W_{new}) = y_{old}
\nonumber
\end{equation}
where $W_{new} = W_{old}+\Delta{w}$.
Briefly, when updates are constrained in the null-space of input features of layer $l$, the outputs at $l$ for previous tasks remain ``\textit{invariant}'' despite the updated parameters.

\section{Problem Statement and Existing Solutions}
\subsection{Problem Formulation}
\label{subsec: LLM Unlearning}
LLM unlearning is defined as the process of selectively removing target knowledge from a pre-trained LLM $\mathcal{M}_\theta$, which is trained on dataset $\mathcal{D}$.
Let $\mathcal D_f \subset \mathcal{D}$ be the target subset of data to be removed, which is referred to as \textit{``forgetting knowledge''} or \textit{``forgetting set.''}
Correspondingly, $\mathcal D_r =\mathcal{D} \backslash \mathcal D_f $ is denoted as retained data, the subset of $\mathcal{D}$ excluding $\mathcal{D}_f$, referred to as \textit{``retaining knowledge''} or \textit{``retaining set''}.
Formally, the objective of LLM unlearning is to derive an unlearned model $\mathcal M_{\theta'}$ whose behaviors are consistent with the model $\mathcal M_{re}^{\mathcal D_r}$ retrained on $\mathcal D_r$:
\begin{equation}
 \mathcal M_{\theta'}(y\mid x) \approx  
 \mathcal  M_{re}^{\mathcal D_r} (y\mid x ),
 \quad \forall (x,y) \in \mathcal{D},
 \nonumber
\end{equation}
where $\mathcal M_{re}^{\mathcal D_r}$ is considered as the \textit{golden baseline} unlearned model. 
Briefly, LLM unlearning is the process of removing knowledge of $\mathcal D_f$ while retaining knowledge of $\mathcal D_r$.

\subsection{Existing Solutions}
\label{subsec: existing solutions}

\mypara{GA~\cite{jang-etal-2023-knowledge}}
\textit{Gradient ascent} (GA) is the most representative method for LLM unlearning.
The general idea is to perform reverse optimization (gradient ascent instead of descent) on $\mathcal D_f$, rewarding LLMs for producing incorrect outputs. 
Concretely, GA can be formulated as \textit{maximizing} the negative log-likelihood loss:
\begin{equation}
\mathcal{L}_{\text{GA}} = \mathbb{E}_{\mathcal D_f}
\left[ \textstyle \sum \log\left(p_\theta(x_t \mid x_{<t}\right) \right], \nonumber
\end{equation}
where $p_\theta(x_t|x_{<t})$ is the prediction probability of token $x_t$ calculated by LLM with parameters $\theta$ and input tokens $x_{<t}$. 

However, directly training with $\mathcal{L}_{GA}$ could cause severe model collapse~\cite{zhangnegative}.
To address this, \textit{gradient difference} (GradDiff)~\cite{mainitofu} introduces a constraint term (referred to as retaining loss $\mathcal{L}_{r}$) on $\mathcal{D}_r$ to prevent excessive forgetting.
The weighted sum constitutes the formal unlearning objective~\cite{yao2024large}:
\begin{equation}
\mathcal{L}_{u} = \mathcal{L}_{f}(x \in \mathcal D_f; \theta) + \alpha \mathcal{L}_{r}(x \in \mathcal D_r; \theta), \nonumber
\end{equation}
where $\mathcal{L}_{f} = \mathcal{L}_{\text{GA}}$, and $\alpha$ is a scaling factor for the balance between forgetting and retention.

\mypara{NPO~\cite{zhangnegative}}
\textit{Negative preference optimization} (NPO) is an alternative to GA.
Inspired by DPO~\cite{rafailov2023direct}, NPO regards the forgetting data as negative samples,
and facilitates forgetting by penalizing the model for generating the dispreferred data:
\begin{equation}
\begin{aligned}
\mathcal{L}_{\text{NPO}}(\theta) = -\frac{2}{\beta}\mathbb{E}_{\mathcal D_f}
\left[\log \sigma \left(-\beta \log \frac{\pi_\theta(y \mid x)}{\pi_{ref}(y \mid x)}\right)\right]
, \nonumber
\end{aligned}
\end{equation}
where $\sigma(\cdot)$ is the sigmoid function, and $\pi_{ref}$ and $\pi_\theta$ denote the original and the unlearned model, respectively. 
In comparison with GA, NPO resolves the absence of an optimization lower bound and exhibits a slower divergence rate, leading to a more stable unlearning process.

\mypara{RMU~\cite{li2024wmdp}}
\textit{Representation misdirection for unlearning} (RMU) is the state-of-the-art method of partial parameter updates.
It suppresses the generation of target knowledge by driving the internal activations of empirically selected layers towards a random uniform distribution, formulated as:
\begin{equation}
\begin{aligned}
\mathcal{L}_{\text{RMU}} = \mathbb{E}_{\mathcal{D}_f} \left[ \| a_f(x) - c \cdot u \|_2^2 \right] 
+ \alpha \mathbb{E}_{\mathcal{D}_r} \left[  \| a_f(x) - a_r(x) \|_2^2 \right]
, \nonumber
\end{aligned}
\end{equation}
where $u$ is a random vector, and $a$ represents the activations. 
$c$ and $\alpha$ denote a forgetting coefficient and a scaling factor.
RMU achieves knowledge removal by disrupting the original activations of forgetting data, while retaining model utility by enforcing activation consistency on retaining data.

\begin{figure*}[!t]
    \centering
    \includegraphics[width=1.0\hsize]{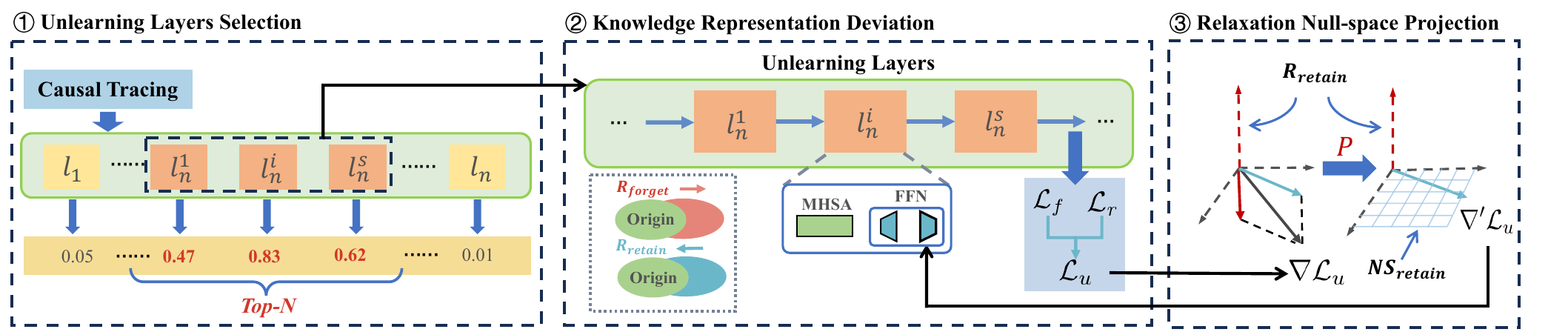}
    \caption{The overview of \sysname pipeline to balance knowledge removal and utility retention for LLMs. \sysname includes three major stages: 1) We employ causal tracing and causal effects to identify FFNs critical for knowledge storage, and apply a sliding window to select the target \textit{unlearning layers} for parameter updates. 
    2) The representations generated from the last unlearning layer are captured and utilized to derive the update gradients $\nabla \mathcal{L}_u$ based on the mechanism of \textit{knowledge representation deviation}. 
    3) The gradients are projected into the \textit{relaxed null-space} of retaining knowledge, and ${\nabla' \mathcal{L}_u}$ is then applied to the last linear transformation matrix of FFN in each unlearning layer.
    $R_{forget}$ and $R_{retain}$ are the representations of forgetting and retaining knowledge. $NS_{retain}$ denotes the relaxed null-space of retaining knowledge. 
    }
    \label{fig:method}
\end{figure*}

\mypara{Drawbacks of Existing Solutions}
The paradigm of full parameter updates exhibits a significant advantage in knowledge removal by deliberately inducing the generation of random responses.
However, given that all parameters are updated indiscriminately, they frequently result in irreversible utility degradation or complete model collapse, and may only achieve response-level mitigation instead of removing knowledge encoded in parameters.
Conversely, partial parameter methods prioritize targeted interventions yet often fail to achieve a satisfactory degree of knowledge removal~\cite{fan2025llmunlearningresilientrelearning,yang2025faithun}.
This deficiency primarily stems from a weak coupling between parameter selection and the knowledge removal algorithm.
While the property of knowledge storage offers a promising design perspective, related efforts remain confined to restricting updates to FFNs' parameters, lacking a synergistic design that aligns the entire pipeline with this property.
For instance, RMU~\cite{li2024wmdp} induces divergence in activations to remove knowledge, yet its parameter selection relies on empirical search.


\section{\sysname}

\subsection{Intuition}
\label{subsec: intuitions}
As detailed in \autoref{subsec: Large Language Model}, FFNs act as \textit{knowledge storage} by 
activating relevant knowledge stored in parameters to influence the generation.
We define the transition process from key to value as ``\textit{knowledge association}'' (\autoref{fig:general ffn}), as this stage retrieves knowledge related to input tokens.
Through this, the representations contain not only the semantic context of the inputs but also their associations with relevant knowledge.

This property enables controllable knowledge removal by selectively perturbing the knowledge associations in specific FFNs to disrupt the generation of target knowledge, while leaving MHSA components intact.
Although MHSA may also contribute to knowledge-related computation, it is closely coupled with contextual aggregation, which supports general language modeling.
Therefore, updating MHSA may broadly degrade contextual processing and model utility.
We verify this phenomenon in Appendix~\ref{subsec:ffn_mhsa}.

While conceptually intuitive, the direct quantification and modification of target knowledge encoded in FFNs remains a formidable challenge.
However, since representations encapsulate both input semantics and their associations with relevant knowledge~\cite{elazar-etal-2021-amnesic,geva2023dissecting}, they offer a more accessible pathway for intervening in knowledge associations.
Consequently, we propose a representation-based unlearning method, named \textit{\textbf{K}nowledge \textbf{U}nlearning by \textbf{D}eviating represent\textbf{A}tion} (\sysname), to disrupt the utilization of target knowledge for generation.

\subsection{Overview}
\label{subsec: Overview}
\sysname aims to leverage the property of knowledge storage in Transformer-based models~\cite{geva2021transformer,meng2022locating} for LLM unlearning. 
To remove knowledge stored in FFNs, \sysname employs selective parameter updates that deviate the representations of target knowledge away from their original states.
The pipeline is illustrated as \autoref{fig:method}.

\mypara{Step 1: Unlearning Layers Selection}
Our approach begins by identifying layers whose FFNs are specialized in knowledge storage.
To this end, we employ causal tracing~\cite{meng2022locating} and define a metric named ``causal effect'' to achieve a component-level identification.
Based on causal effect, we propose a sliding window search strategy to select unlearning layers, a portion of whose parameters are the targets for subsequent updates.
We refer readers to \autoref{subsec: Layers selection} for more details.

\mypara{Step 2: Knowledge Representations Deviation}
Knowledge stored in FFN parameters is retrieved and incorporated into representations through knowledge association.
Building upon this, we design a knowledge removal loss, which involves inducing the representation of target knowledge away from its original state. 
This deviation perturbs the established knowledge associations, thereby disrupting the generation of target knowledge.
Moreover, a constraint loss is also introduced to limit the changes in the representation space.
The detailed design of unlearning objectives is presented in \autoref{subsec: knowledge representation deviation}.

\mypara{Step 3: Null-space Projection}
To address directional conflicts between forgetting and retention, we introduce the null-space to filter out the components of the unlearning gradient derived in Step 2 that would degrade model utility.
To prevent the null-space from excessively suppressing the knowledge removal, we further introduce a relaxation operation, yielding a relaxed null-space projection.
Besides, we propose stochastic proportional sampling for the computational bottleneck. 
The projected gradients preserve utility without compromising the effectiveness of forgetting.
See \autoref{subsec: null-space projection} for more details.

\subsection{Unlearning Layers Selection}
\label{subsec: Layers selection}
\sysname begins by identifying the LLM layers whose FFNs are predominantly responsible for storing knowledge, thereby guiding the selection of parameters to update.

\mypara{Causal Effect}
Previous studies on layer selection for unlearning primarily focused on either the entanglement degree between forgetting and retaining knowledge~\cite{pochinkov2024dissectinglanguagemodelsmachine}, or on the layer's association with the LLMs' refusal behaviors~\cite{shen2025lunar}.
However, these coarse-grained, layer-level methods overlook the functional distinctions among the components within a layer.
Concretely, the effects of MHSA and FFN are merged through the residual connection, making the results of layer selection inherently influenced by both components simultaneously.
This indicates that the selected layers may be critical for semantics processing dominated by MHSA~\cite{jawahar2019bertlearn,hanna2023does}, instead of those driven by FFNs for knowledge storage (see \autoref{subsec: Large Language Model} for more details).
Consequently, such coarse-grained selection may inadvertently impair the models' capabilities for context awareness or text generation.

To exploit the knowledge storage property of LLMs, we propose a fine-grained, component-level identification method.
Specifically, we identify layers whose FFNs exhibit specialization in knowledge storage, thereby decoupling the roles of MHSA and FFN during layer selection.
This motivates the adoption of \textit{causal tracing}, which quantifies the contribution of each intermediate activation, as detailed in \autoref{subsec: Causal Tracing}.
After performing causal tracing, we define the average probability difference between ``Corrupted Run'' and ``Corrupted-with-restoration Run'' as \textit{Causal Effect} ($\mathrm{CE}^l_{m}$), which measures the importance of components $m$ at layer $l$ for correct generation:
\begin{equation}
\mathrm{CE}^l_{m}=\frac{1}{B\cdot T} \sum^B  \sum^T_{i=1} (\mathbb{P}^*_{m^l_i} - \mathbb{P}^*),
\end{equation}
where $B$ and $T$ represent the input batch size and sequence length. 
Concretely, a higher $\mathrm{CE}^l_{FFN}$ score indicates that the FFN at layer $l$ makes a significant contribution in storing and retrieving knowledge during generation.

Ideally, causal tracing should be performed on specific datasets (\eg, $\mathcal{D}_f$ or $\mathcal{D}_r$) to pinpoint FFNs storing target knowledge.
However, as discussed in \autoref{subsec: efficiency}, the substantial computational overhead of causal tracing renders it impractical to recompute CE for different unlearning datasets.
Notably, our analysis in Appendix~\ref{subsec: causal effect of various datasets} shows that the distribution of $\mathrm{CE}^l_{FFN}$ is consistent across datasets.
Motivated by this observation, we execute causal tracing \textit{only once per model} using ``1,000 factual statements~\cite{meng2022locating}.''
Since this dataset spans a wide range of knowledge domains, the generated CE scores exhibit generalizability across topics and thus can be reused to mitigate the need for causal tracing in future unlearning tasks.

\mypara{Selection Strategy}
Based on empirical investigation, we observe that though with high $\text{CE}^l_{FFN}$ scores, \sysname is sensitive to the layer positions.
Concretely, the first several layers of the FFN-dominant zone primarily handle generic features~\cite{jawahar2019bertlearn}, while those layers located at its terminus facilitate the functional transition from knowledge retrieval to semantic aggregation, rendering both of them ill-suited as unlearning layers.
Besides, contiguous layers should be prioritized to preserve the coherence of the representation space across layers.
See Appendix~\ref{subsec: unlearning layers analysis} for more detailed analysis.

Therefore, we design a \textit{sliding window search} strategy to select unlearning layers, as presented in \autoref{alg: sliding window search}.
First, we construct a candidate set $\mathcal{S}_{top}$ containing $N$ layers with the highest $\text{CE}^l_{FFN}$ scores, ordered by layer indices.
Starting from the median layer of $\mathcal{S}_{top}$ and iterating toward the deeper layers, we treat a candidate layer $l_c$ as a central anchor to bidirectionally expand a sliding window $\mathcal{W}_{c,s}$ of size $s$.
The layers within $\mathcal{W}_{c,s}$ are selected as \textit{Unlearning Layers} only if $\mathcal{W}_{c,s}$ meets the Hit Ratio (HR) condition:
\begin{equation}
    \text{HR}(\mathcal{W}_{c,s}) = \frac{|\mathcal{W}_{c,s} \cap \mathcal{S}_{top}|}{s} > 0.5
\end{equation}
With a majority-principle threshold of $0.5$, the HR constraint ensures that the window is mainly composed of layers with significant contributions to knowledge storage, filtering out isolated outliers.

\begin{algorithm}[!t]
\caption{Sliding Window Search}
\label{alg: sliding window search}
\begin{algorithmic}[1]
\Require Scores $\{\text{CE}^l_{FFN}\}$, candidate size $N$, window size $s$
\Ensure Selected unlearning layers $\mathcal{W}_{un}$
\State Select the top-$N$ layers according to $\mathrm{CE}_{FFN}^l$
\State $\mathcal{S}_{top} = [l_{(1)}, \cdots, l_{(N)}]$, where ${(1)} \prec \cdots \prec {(N)}$
\State $p \gets \lceil s/2 \rceil$ \Comment{Boundary offset}
\For{$i = \lceil (N+1)/2 \rceil$ \textbf{to} $N$} 
    \State $l_{c} \gets \mathcal{S}_{top}{[i]}$  \Comment{Central anchor}
    \State $\mathcal{W}_{c,s} \gets \{l_{c-p+1}, \cdots,l_{c-1},l_{c}, l_{c+1},\cdots, l_{c+s-p}\}$ 
    
    \If{$ \text{HR}(\mathcal{W}_{c,s}) > 0.5$} \Comment{Hit Ratio condition}
        \State \Return $\mathcal{W}_{un} \gets \mathcal{W}_{c,s}$
    \EndIf
\EndFor
\State \Return $\mathcal{W}_{un} \gets \mathcal{W}_{\lceil (N+1)/2 \rceil,s}$
\end{algorithmic}
\end{algorithm}

\subsection{Knowledge Representation Deviation}
\label{subsec: knowledge representation deviation}

As discussed in \autoref{subsec: intuitions}, input tokens activate knowledge stored in specific FFNs during generation, establishing associations that are then incorporated, along with the semantic context of inputs, into the hidden representations~\cite{geva2021transformer,meng2022locating}.
For instance, given the input ``Apple CEO,'' the representations of specific layers might contain the association with ``Tim Cook,''  which is the implicit knowledge stored in FFNs and triggered by the input context.
Such knowledge associations refine the information in representations and are subsequently extracted by MHSA to influence the final output~\cite{geva2023dissecting}.
Therefore, disrupting these associations in representations can effectively inhibit the generation of target knowledge, achieving the goal of unlearning.
To this end, we design the representation-based objectives for target knowledge removal.

\mypara{Forgetting Loss}
Inspired by preference optimization~\cite{rafailov2023direct,zhangnegative}, we design a deviation mechanism to push the representation of forgetting knowledge away from its original state. 
Once the hidden representations are perturbed from their pre-trained state, the incorporated associations with relevant knowledge can no longer be effectively utilized by subsequent layers.

While preference optimization relies on the similarity of model outputs, the lack of a direct mapping between hidden representations and outputs calls for a substitute metric to quantify the similarity of representations.
To address this issue, we adopt \textit{cosine similarity} between representations of the unlearning model and the original ones. 
Given the high dimensionality and sparsity of representations, cosine similarity is particularly well-suited, as it measures direction consistency rather than magnitude and is robust to sparse activation patterns.
Besides, its value transition of $1 \xrightarrow{} 0$ provides a mathematical implication of ``reducing similarity.''
This process effectively drives the representation away from its original state,  intuitively aligning with the goal of preference optimization.
Furthermore, we adopt the cosine complement $(1-\cos)$, which enables similarity reduction while simultaneously ensuring the training objective remains convergent.
Accordingly, we refine preference optimization by integrating the aforementioned designs to construct the forgetting objective, which induces deliberate deviation in the representations of forgetting knowledge:
\begin{equation}
\begin{gathered}
\mathrm{d}_f(x) = \cos \langle \mathcal{R}_{un}^{l'}(x),\mathcal{R}_{ori}^{l'}(x) \rangle \\
\mathcal L_f = -\frac{2}{\beta}\mathbb{E}_{\mathcal D_f}
\left[
\log \sigma 
\left( 
\beta \cdot \left( 1-\mathrm{d}_f(x) \right) 
\right) 
\right],
\end{gathered}
\end{equation}
where $x \in \mathcal D_f$ are the forgetting data, $\cos\langle\cdot\rangle$ denotes the cosine similarity, and $\sigma(\cdot)$ refers to the sigmoid function. 
The hyperparameter $\beta$ represents the inverse unlearning temperature.
Notably, $\mathcal{R}_{un}^{l'}$ and $\mathcal{R}_{ori}^{l'}$ denote the representations at the \textit{last unlearning layer} $l' \in \mathcal{W}_{un}$ from the unlearning and the original models. 
Since knowledge is distributed across layers and aggregated through residual connections, the representation at layer $l'$ integrates a substantial amount of forgetting knowledge and corresponding associations. 
This makes it a principled choice as the target in $\mathcal{L}_f$.

Briefly, the minimization of $\mathcal{L}_f$ serves to reduce similarity $\cos\langle\cdot\rangle$ between the current and original representations. 
This derives a deviation of the target knowledge representation from its pre-trained state, which inhibits the subsequent generation of that knowledge.

\mypara{Retaining Loss}
Deviating target representations can inadvertently impact unrelated knowledge.
To restrict the scope of the forgetting effect, we introduce an additional constraint term.
Classical output-level constraints, which employ KL-divergence or cross-entropy to restrict token distribution changes, cannot be directly adopted to deal with hidden representations. 
The reason is that the representations reside in a continuous vector space rather than a discrete probability space.
Therefore, we adopt an $\ell_2$-regularization term on $\mathcal D_r$ to effectively penalize excessive deviations within the latent representation space: 
\begin{equation}
\mathcal L_r = \mathbb{E}_{\mathcal D_r}
\left[ 
\scalebox{1.3}{$\ell$}_2 \big( \mathcal{R}_{un}^{l'}(x),\mathcal{R}_{ori}^{l'}(x) \big)
\right],
\end{equation}
where $x \in \mathcal D_r$ are the retaining data, and other symbols are consistent with those defined in $\mathcal L_f$.
Beyond preserving factual knowledge, $\mathcal D_r$ is also essential for retaining the structural knowledge, such as grammatical and syntactic patterns, that form the foundation of excellent linguistic capabilities~\cite{ren2025sokmachineunlearninglarge}.

\mypara{Unlearning Loss}
The final unlearning loss is formulated as the composite of forgetting and retaining objectives, which are responsible for knowledge removal and utility retention  :
\begin{equation}
\mathcal L_u = \mathcal L_f + \mathcal L_r
\end{equation}
In contrast to conventional methods, we discard the weighted unlearning objective (\ie, $\mathcal L_f + \alpha \cdot \mathcal L_r$), which typically employs a scaling factor to balance forgetting and retention. 
In practice, due to the difference in loss functions and task requirements, the tuning effect of such hyperparameters might be non-linear and substantially increase the cost of hyperparameter searching in specific tasks. 
This complexity, in turn, diminishes the adaptability across diverse unlearning scenarios.
We will elaborate on its fundamental limitation and our improvement in \autoref{subsec: null-space projection}.

\mypara{Parameter Update}
According to \autoref{subsec: Large Language Model}, the \textit{last linear transformation matrix $W_{out}$} (referred to as $W_{down}$ in some LLMs, such as Llama) in FFNs plays a crucial role in knowledge association. 
Therefore, we perform unlearning by applying parameter updates $\nabla_w\mathcal{L}_u$ exclusively to $W_{out}$ of FFNs in \textit{all unlearning layers}, since knowledge is distributed across multiple layers~\cite{geva2021transformer}.
This selective refinement enables precise removal of target knowledge while preserving model utility. 

\subsection{Relaxed Null-Space Projection}
\label{subsec: null-space projection}

As shown in \autoref{subsec: null-space}, null-space has the property of \textit{invariance}, motivating us to adopt this in reconciling the conflicts between two unlearning targets.

\mypara{Drawbacks of Scaling Factor}
The most representative paradigm for utility retention is to incorporate the retaining loss and employ a scaling factor to achieve the trade-off between forgetting and retention, \ie, $\mathcal L_f + \alpha \cdot \mathcal L_r$.
However, previous studies~\cite{wangrethinking,wang2025gru,zhangnegative} have revealed its limited effectiveness in preventing utility degradation, due to the failure to reconcile the conflict between forgetting and retention.
Specifically, any component of the unlearning gradients that acts in opposition to the retention direction will inevitably degrade model utility.
This stems from a fundamentally \textit{directional} conflict that cannot be resolved by merely adjusting the magnitude of gradients through the scaling factor $\alpha$.

These limitations motivate a subspace transformation mechanism to eliminate the utility-harming components of the unlearning gradient.
The invariance property of null-space provides a solution.
Concretely, we project gradients $\nabla_{w^l} \mathcal{L}_u$ into the null-space of the input feature of $\mathcal D_r$ at layer $l$.
This remains the representations of retaining knowledge \textit{invariant} despite updated parameters, as introduced in \autoref{subsec: null-space}.

\mypara{Input Feature Capturing} 
We first approximate the input feature space of retaining knowledge at $W_{out}$ using the uncentered covariance statistic of $k$, where $k$ denotes the input of $W_{out}$ as defined in \autoref{subsec: Large Language Model}.
The covariance captures dominant correlations across feature dimensions. 
Concretely, we feed retaining data $x_r \in \mathcal D_r$ into the original model to collect $k^l_r$ from all unlearning layers.
Based on this, the uncentered covariance statistic $\mathcal{K}^l_r \in \mathbb{R}^{\mathrm{d} \times \mathrm{d}}$ is denoted as the input feature space and formulated as:
\begin{equation}
    \mathcal{K}^l_r =  \frac{1}{n_r} \sum \nolimits^{n_r}_{i=1}(k^l_{r,i})^\top k^l_{r,i}
\end{equation}

However, this phase requires $\left| \mathcal D_r \right|$ times of forward passes to collect the input features, which limits scalability in practice for large-scale retaining datasets.
To this end, we propose \textit{Stochastic Proportional Sampling}, which constructs an approximate input feature space using a sampled retaining subset $\hat{\mathcal{D}}_r \subseteq \mathcal{D}_r$:
\begin{equation}
\begin{gathered}
\hat{\mathcal{D}}_r
=
\{x_{r_i} \mid i \sim \mathrm{Uniform}(1, |\mathcal{D}_r|)\} \\
\text{ s.t. }|\hat{\mathcal{D}}_r| = p\% \cdot |\mathcal{D}_r|, \\
\end{gathered}
\end{equation}
where $p$ is the sampling ratio.
This stochasticity ensures the statistical representativeness of the captured feature distribution, while the proportional sampling effectively reduces computational overhead by reducing the number of forward passes.

\mypara{Null-Space Relaxation}
Next, we apply \textit{singular value decomposition} (SVD) to $\mathcal{K}^l_r$ to identify feature directions that are uncorrelated with retaining knowledge:
\begin{equation}
U^l, \Lambda^l, (U^l)^\top = \mathrm{SVD}(\mathcal{K}^l_r),
\end{equation}
where the subset of singular vectors $U^l_\mathrm{z} \subset U^l$, corresponding to zero singular values, constitute the \textit{strict null-space}~\cite{wang2021training}.
The gradients are then projected into the null-space of retaining knowledge with the \textit{null-space projection matrix} ${U_\mathrm{z}}^l ({U_\mathrm{z}^l})^\top$~\cite{wang2021training}.
However, given the representation \textit{entanglement} between the forgetting and retaining knowledge, a strict projection matrix inadvertently constrains the updates to representations of forgetting knowledge, thereby hindering the thorough removal.

To this end, we introduce a \textit{null-space relaxation threshold} $\tau$ to relax the strict constraint of zero-singular-value.
Specifically, we construct an approximate null-space by selecting the singular vectors whose corresponding singular values fall below $\tau$:
\begin{equation}
    \mathcal{N}^l_\tau = \left\{ u^l_i \mid \lambda^l_i < \tau, \lambda^l_i \in \Lambda,  u^l_i \in U^l \right\},
\end{equation}
Based on this, the \textit{relaxed null-space projection matrix} is then defined as $P^l_\tau = \mathcal{N}^l_\tau (\mathcal{N}^l_\tau)^\top \in \mathbb{R}^{\mathrm{d} \times \mathrm{d}}$.
This constrains gradients such that updates occur along directions weakly correlated with retaining knowledge, thus removing target knowledge without incurring significant degradation to utility.

\mypara{Gradients Projection}
Finally, before updating parameters, we project the gradients into the relaxed null-space of retaining knowledge with $P_{\tau}^l$, ensuring that the update gradients do not impair the content generation related to retaining knowledge:
\begin{equation}
\nabla'_{w^l} \mathcal{L}_u = P^l_{\tau} \cdot \nabla_{w^l} \mathcal{L}_u  
\end{equation}
where $\nabla'_{w^l} L_u$ is adopted to update $W_{out}^l$ at layer $l$.

\subsection{Principled Hyperparameter Configuration}
\label{subsec: principled hyperparameter configuration}
The performance of \sysname is mainly influenced by the inverse unlearning temperature ($\beta$) and the null-space relaxation threshold ($\tau$).
$\beta$ controls the intensity of forgetting and enhances knowledge removal when decreased.
A smaller $\tau$ makes the null-space less correlated with retaining knowledge, thereby enhancing utility retention while weakening forgetting.
As such, we propose a \textit{two-stage tuning} strategy for practical deployment:
\begin{enumerate}[label=\arabic*), leftmargin=*, nosep]
    \item \textbf{Stability Boundary Identification.}
    We fix $\beta=0.1$ to enforce a state of excessive forgetting.
    Under this setting, we decrease $\tau$ and monitor an overall forgetting metric, which is later formalized as KRD in \autoref{eq: KRD}.
    As shown in \autoref{fig: tau boundary sample}, we set $0.9$ as a reference and define the \textit{stability boundary} as the first $\tau$ at which the forgetting metric KRD falls below the reference.
    Further detailed analyses are provided in Appendix~\ref{subsec: stability boundary}.
    \item \textbf{Unlearning Balance.}
    Based on the identified boundary, we select a $\tau$ slightly \textit{larger} than it to ensure effective removal.
    For instance, the boundary in \autoref{fig: tau boundary sample} is $1 \times 10^{-3}$, leading to the choice of $\tau = 2 \times 10^{-3}$.
    Then, we increase $\beta$ across $\{ 0.1, 1.0, 2.0\}$. 
    This process mitigates forgetting intensity to balance knowledge removal and utility retention.
\end{enumerate}
The principled strategy prioritizes the stability of the representation space under an extreme forgetting scenario and then refines the unlearning balance, thereby transforming a complex joint optimization into two sequential one-dimensional searches.

\begin{figure}[!t]
    \centering
    \includegraphics[width=1.0\hsize]{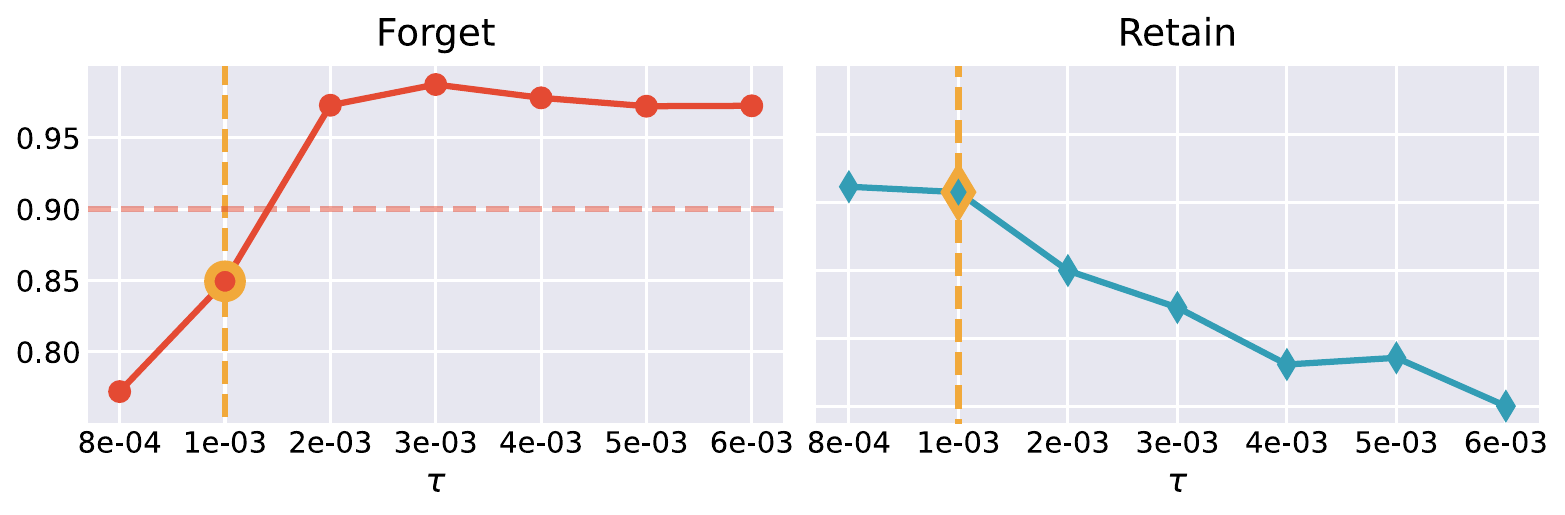}
    \caption{
    Unlearning with fixed $\beta=0.1$ and different $\tau$.
    The two plots show the forgetting and retention performance.
    A reference line at $0.9$ is set for the forgetting metric.
    The highlighted $\tau$ denotes the \textit{stability boundary}, defined as the first $\tau$ where the forgetting metric falls below the reference.
    }
    \label{fig: tau boundary sample}
\end{figure}

\section{Evaluation}
\label{sec:evaluation}

\subsection{Experimental Setup}
\label{subsec:experimental setup}

There are multiple well-established benchmarks for evaluating the performance of various LLM unlearning algorithms.
These benchmarks incorporate a set of datasets, metrics, and models.
In this paper, we evaluate the performance of \sysname on two representative benchmarks, MUSE~\cite{shi2024muse} and WMDP~\cite{li2024wmdp}.

\mypara{Datasets}
The datasets consist of \textit{forgetting set} $\mathcal{D}_f$ and \textit{retaining set} $\mathcal{D}_r$. 
\begin{itemize}
    \item \textbf{MUSE.} 
    It consists of two datasets, BOOKS and NEWS. 
    For BOOKS, the forgetting set comprises text excerpts from the Harry Potter series, while the retaining set is sampled from the Harry Potter Fandom Wiki. 
    For NEWS, the forgetting and retaining sets are constructed by randomly partitioning BBC News articles published after August 2023. 
    
    \item \textbf{WMDP.} 
    It focuses on the malicious misuse of hazardous knowledge in LLMs, aiming to evaluate the ability to remove dangerous knowledge from models. 
    Its forgetting set comprises harmful knowledge in Biosecurity (Bio) and Cybersecurity (Cyber), collected from GitHub and PubMed, respectively.
    The retaining set is from Wikitext.
\end{itemize}

\begin{table*}[]
\caption{
The comprehensive unlearning performance comparison between ours and baselines on MUSE, evaluating the forgetting quality and retaining utility from multiple dimensions. 
KRD represents the overall effectiveness of forgetting, and KnowMem in $\mathcal D_r$ reflects the retention performance.
The best results in Forgetting Quality and Retaining Utility are highlighted in \textit{\textbf{bold}}, and the suboptimal ones are marked with an \underline{underline}.
}
\label{tab:muse performance}
\setlength{\tabcolsep}{0.65em} 
\renewcommand{\arraystretch}{1.1}
\footnotesize
\centering
\begin{tabular}{c ccccc ccccc}
\toprule[1pt]
\multirow{3}{*}{\textbf{Unlearning Method}} & \multicolumn{4}{c}{\textbf{Forgetting Quality}} & \textbf{Retaining Utility} & \multicolumn{4}{c}{\textbf{Forgetting Quality}} & \textbf{Retaining Utility} \\ \cmidrule(lr){2-5} \cmidrule(lr){6-6} \cmidrule(lr){7-10} \cmidrule(l){11-11}
                                 &
                                 \begin{tabular}[c]{@{}c@{}}{VerbMem}\\ $\mathcal D_f$$\downarrow$\end{tabular} 
                                 & \begin{tabular}[c]{@{}c@{}}{KnowMem}\\ $\mathcal D_f$$\downarrow$\end{tabular}
                                 & \begin{tabular}[c]{@{}c@{}}{PrivLeak}\\ $\rightarrow{} 0$\end{tabular}
                                 & \begin{tabular}[c]{@{}c@{}}{\textbf{KRD}}\\ $\boldsymbol{\rightarrow{} 1}$\end{tabular}
                                 & \begin{tabular}[c]{@{}c@{}}{\textbf{KnowMem}}\\ $\boldsymbol{\mathcal D_r \uparrow}$\end{tabular}       
                                 & \begin{tabular}[c]{@{}c@{}}{VerbMem}\\ $\mathcal D_f$$\downarrow$\end{tabular} 
                                 & \begin{tabular}[c]{@{}c@{}}{KnowMem}\\ $\mathcal D_f$$\downarrow$\end{tabular}
                                 & \begin{tabular}[c]{@{}c@{}}{PrivLeak}\\ $\rightarrow{} 0$\end{tabular}
                                 & \begin{tabular}[c]{@{}c@{}}{\textbf{KRD}}\\ $\boldsymbol{\rightarrow{} 1}$\end{tabular}
                                 & \begin{tabular}[c]{@{}c@{}}{\textbf{KnowMem}}\\ $\boldsymbol{\mathcal D_r \uparrow}$\end{tabular}           \\ \midrule
                                 & 
                                 \multicolumn{5}{c}{\textbf{BOOKS}} & \multicolumn{5}{c}{\textbf{NEWS}} \\
\cmidrule(lr){2-6} \cmidrule(l){7-11}
Origin                           & 99.59                & 58.76                 & 56.68                & 0.0000                    & 67.00                    & 58.42                & 63.41                 & 99.84                & 0.0000                    & 54.96                    \\
Retrain                          & 14.35                & 28.90                 & 0.00                  & 1.0000                    & 74.38                    & 20.80                & 33.30                 & 0.00                  & 1.0000                    & 54.68                    \\ \midrule
GA                            & 11.00                & 23.96                 & 40.50                & 0.5452                    & 34.24                    & 4.95                 & 31.08                 & 108.10                & 0.0002                    & 27.33                    \\
NPO                           & 20.13                & 27.68                 & 41.35                & 0.5200                    & 44.75                    & 7.60                 & 54.62                 & 105.79                & 0.0003                    & {\ul 41.27}                    \\
WHP                              & 19.04                & 51.49                 & 51.56                 & 0.1849                    & \textbf{63.47}                    & 17.15                & 23.17                 & 97.56                 & 0.0658                    & 29.90                    \\
SimNPO                           & 0.00                 & 3.24                  & 17.15                & {\ul 0.8737}                    & 49.55                    & 13.21                & 47.82                 & 11.75                 & {\ul 0.7380}                    & 40.66                    \\
RMU                              & 7.04                 & 30.92                 & 20.45                & 0.8248                    & 58.37                    & 18.14                & 32.31                 & 93.77                & 0.1627                    & 38.95                    \\ 

\rowcolor[RGB]{227,234,241}\textbf{\sysname}                         & 0.05                 & 26.73                 & 16.55                & \textbf{0.8792}                    & {\ul 61.97}                    & 8.71                 & 0.26                  & 20.09                 & \textbf{0.9225}                   & \textbf{53.08}                    \\ \bottomrule[1pt]
\end{tabular}
\end{table*}

\mypara{Models}
We need two models to evaluate the performance of unlearning methods: 
The ``origin'' model is the starting point of the unlearning process and serves as the performance ceiling for retention.
The ``retrain'' model denotes the golden baseline and provides an ideal and expected reference for forgetting.

\begin{itemize}
    \item \textbf{MUSE.} 
    It adopts ICLM‑7B~\cite{shi2023context} and Llama‑2‑7B as ``retrain'' models for the BOOKS and NEWS, respectively. 
    Since neither of their pre-training corpora includes MUSE datasets, they can be approximately considered equivalent to retrained. 
    These models are then fine-tuned on MUSE datasets to obtain ``origin'' models, which are released on Hugging Face.

    \item \textbf{WMDP.} 
    It adopts Zephyr-7B-Beta as the ``origin'' model, and we extend to modern ones, \eg, Llama‑3.1‑8B and Qwen3-8B. 
    A ``retrain'' model is unavailable due to the widespread distribution of hazardous knowledge across training corpora.
\end{itemize}

\mypara{Metrics}
Different benchmarks design their own metrics based on the unique characteristics of datasets.
These metrics can be broadly classified into two categories: \textit{forgetting quality} (FQ) and \textit{retaining utility} (RU). 
Due to the space limitation, we defer the formal definition of these metrics to Appendix~\ref{subsec: Metrics}.
\begin{itemize}
    \item \textbf{MUSE.} 
    For \textit{FQ}, MUSE designs three metrics (\textit{VerbMem}, \textit{KnowMem} on $\mathcal D_f$, and \textit{PrivLeak}). 
    The first two metrics reflect the ability to reproduce training data and to answer questions related to forgetting knowledge.
    When values are lower than those of the ``retrain'' model, forgetting is considered completed.
    Besides, a value closer to $0$ is preferred for PrivLeak, reflecting good defense against membership inference attacks (MIA). 
    Regarding \textit{RU}, MUSE utilizes \textit{KnowMem} on $\mathcal D_r$, and the value closer to the ``origin'' model indicates better retention. 
    
    \item \textbf{WMDP.} 
    It assesses \textit{FQ} with \textit{Answer Acc} (Acc) on multi-choice test sets, including biology and cybersecurity. 
    $25\%$ is the ideal result, indicating that the model lacks cognition of specific knowledge and thus answers randomly.
    Besides, WMDP adopts \textit{MMLU}~\cite{hendrycks2020measuring}, a multi-domain benchmark covering STEM and humanities, to assess RU of broad world knowledge.
    The score closer to that of the ``origin'' model reflects better retention of general knowledge.  
\end{itemize}
RU is evaluated via a single metric, whereas FQ is characterized by multiple metrics, making it difficult to compare the \textit{forgetting} performance across methods.
Therefore, we propose \textit{knowledge removal degree (KRD)}, an aggregation of forgetting metrics:
\begin{equation}
    \mathrm{KRD}=\frac{n}{\sum^n_{i=1} {\frac{1}{\mathcal{FQ}^i}}},
    \label{eq: KRD}
\end{equation}
where $\mathcal{FQ}^i$ denotes the $i$-th forgetting metric in the specific benchmark (\eg, VerbMem in MUSE), and is adjusted to avoid rewarding excessive forgetting.
The adjustment process is detailed in Appendix~\ref{subsec: Metrics}.
When any metric shows poor forgetting, KRD tends towards 0; otherwise, it approaches 1.
A favorable forget-retain trade-off is therefore indicated when KRD is close to 1 while the retention metric does not substantially degrade.
Notably, KRD is designed for comparing different methods within the same benchmark.

\mypara{Competitors}
We compare with a series of representative methods: GA~\cite{jang-etal-2023-knowledge}, NPO~\cite{zhangnegative}, SimNPO~\cite{fan2024simplicity}, and the SOTA method RMU~\cite{li2024wmdp}. 
All methods are implemented with GradDiff (GD)~\cite{yao2024large}. 
For other baselines, we conduct WHP~\cite{eldan2310s} in MUSE.
See Appendix~\ref{subsec: Baselines Implementation} for more implementation details.

\subsection{Overall Unlearning Performance}
\label{subsec: overall performance}
We evaluate \sysname on MUSE and WMDP, and quantify unlearning performance in terms of forgetting quality and retaining utility.

\mypara{Performance on MUSE}
\autoref{tab:muse performance} presents the results on MUSE. 
In general, \sysname demonstrates outstanding effectiveness of knowledge removal in both BOOKS and NEWS. 
The optimal KRD values indicate comprehensive forgetting capabilities, ranging from the deletion of verbatim memorization to privacy protection.
Concretely, \sysname significantly reduces both VerbMem and KnowMem on $\mathcal D_f$ below those of the ``retrain'' model, suggesting that the LLM’s capacity to utilize the target forgetting knowledge for generation has been effectively suppressed.
In contrast, WHP, NPO, and SimNPO only weaken the model's ability to reproduce specific training text, reflecting superficial forgetting.
Additionally, it can be observed that \sysname also shows great improvement in PrivLeak, only slightly underperforming SimNPO in NEWS. 
This reveals that the model unlearned by \sysname can effectively defend against MIA, demonstrating solid capabilities of privacy protection.

Furthermore, \sysname achieves a great balance between forgetting effectiveness and model utility retention.
In NEWS, \sysname minimizes degradation in retaining utility while achieving maximal knowledge removal.
In BOOKS, although slightly weaker than WHP in retention, \sysname presents significant forgetting effectiveness, indicating a more balanced unlearning performance;
In comparison, WHP yields a model with a state close to ``not forgetting.''
Notably, GA thoroughly removes target knowledge via excessive forgetting~\cite{shi2024muse}, but at the cost of severely compromising model utility and even increasing vulnerability to MIA.
From this perspective, \sysname reveals that its unlearning mechanism dives into knowledge encoded in LLMs, without excessive destruction to model utility.

Moreover, RMU presents strong performance in both forgetting and retention. 
However, \sysname excels further in dataset generalization, comprehensive knowledge removal, and utility preservation.

\begin{table}[!t]
\caption{The unlearning performance comparison between ours and baselines on WMDP. 
Acc-Bio and Acc-Cyber denote the answer accuracy on different test sets. 
Acc-Avg denotes the weighted average accuracy, with the weights determined by the number of questions in each set. 
Other definitions are consistent with those of \autoref{tab:muse performance}.
}
\label{tab:wmdp performance}
    \setlength{\tabcolsep}{0.35em}
    \renewcommand{\arraystretch}{1.1}
    \footnotesize
    \centering
\begin{tabular}{cccccc}
\toprule[1pt]
\multicolumn{1}{c}{\multirow{2}{*}{\textbf{Method}}} & \multicolumn{4}{c}{\textbf{Forgetting Quality}}                                      & \textbf{Retaining Utility} \\ \cmidrule(lr){2-5} \cmidrule(l){6-6} 
\multicolumn{1}{c}{}                        & Acc-Bio$\downarrow$         & Acc-Cyber$\downarrow$             & Acc-Avg$\downarrow$              & \textbf{KRD}$\boldsymbol{\rightarrow{} 1}$ & \textbf{MMLU}$\boldsymbol{\uparrow}$                     \\ \midrule
\multicolumn{6}{c}{\textbf{Zephyr-7B-Beta}}                                                                                                                                              \\ \midrule
\multicolumn{1}{c}{Origin}                  & 63.82           & 43.78                 & 51.60                 & 0.0000                    & 58.13                    \\ \midrule
\multicolumn{1}{c}{GA}  & 27.06    & 26.82           & 26.92           & \textbf{0.9244}                    & 33.02                    \\
\multicolumn{1}{c}{NPO} & 42.93           & 36.34                 & 38.91                 & 0.6813                    & 45.71                    \\
\multicolumn{1}{c}{SimNPO}                  & 39.61           & 33.79                 & 36.06                 & 0.5741                    & 48.26                    \\
\multicolumn{1}{c}{RMU}                     & 29.22           & 27.93                 & 28.43                 & 0.8665                    & {\ul 56.90}              \\
\rowcolor[RGB]{227,234,241}\multicolumn{1}{c}{\textbf{\sysname}}                 & 29.92 & 26.82        & 28.03       & {\ul 0.8879}           & \textbf{57.14}           \\ \midrule
\multicolumn{6}{c}{\textbf{Llama-3.1-8B}}                                                                                                                                                  \\ \midrule
\multicolumn{1}{c}{Origin}                  & 59.62           & 41.87                 & 48.80                 & 0.0000                    & 63.31                    \\ \midrule
\multicolumn{1}{c}{RMU}                     & 30.00    & 26.27                 & 27.73                 & 0.8888                    & 58.33                    \\
\rowcolor[RGB]{227,234,241}\multicolumn{1}{c}{\textbf{\sysname}}                 & 29.22 & 25.00 & 26.28 & \textbf{0.9350}           & \textbf{60.64}    \\ \midrule
\multicolumn{6}{c}{\textbf{Qwen-3-8B}}                                                                                                                                                    \\ \midrule
\multicolumn{1}{c}{Origin}                  & 74.16           & 45.14                 & 56.47                 & 0.0000                    & 73.01                    \\ \midrule
\multicolumn{1}{c}{RMU}                     & 36.52    & 34.35                 & 35.30                 & 0.6304                    & 60.02                    \\
\rowcolor[RGB]{227,234,241}\multicolumn{1}{c}{\textbf{\sysname}}                 & 28.43 & 26.17 & 27.05 & \textbf{0.9360}           & \textbf{68.92}    \\ \bottomrule[1pt]
\end{tabular}
\end{table}

\mypara{Performance on WMDP}
In \autoref{tab:wmdp performance}, we present the performance on WMDP. 
In general, \sysname achieves both effective forgetting and strong utility retention. 
It removes $88.8\%$ of the hazardous knowledge stored in Zephyr-7B-Beta, with only a negligible $1.3\%$ drop in retaining utility.
Notably, it maintains better performance than RMU, even though RMU is the SOTA method on this benchmark.
As for other baselines, GA leads to severe degradation in retaining utility for achieving forgetting. 
NPO and SimNPO struggle to eliminate the model's ability to utilize harmful knowledge, despite SimNPO's strong performance on MUSE.
Both methods achieve forgetting by suppressing target-response generation, failing to remove the underlying harmful knowledge.
In contrast, \sysname achieves multi-dimensional unlearning.

Furthermore, since the ``origin'' model for WMDP does not require dataset-specific fine-tuning, we extend our experiments to more modern architectures, Llama-3.1 and Qwen-3.
In subsequent experiments, we only compare with RMU, as both \sysname and RMU substantially outperform all other baselines.

Overall, \sysname consistently surpasses RMU, particularly on Qwen-3.
Both the Zephyr and Llama series originate from the Llama architecture family, sharing high similarity in model design and compatibility. 
In contrast, Qwen-3 incorporates architectural innovations such as QK-Norm and a deeper stack of layers. 
We attribute RMU’s limited transferability primarily to poor compatibility with these architectural differences.
This demonstrates \sysname's generalization across models.

\mypara{Performance on TOFU}
We further conduct experiments on TOFU~\cite{mainitofu}, and the detailed experiments are presented in Appendix \ref{sec: tofu} due to space limitations.
The results show that \sysname achieves complete forgetting without utility degradation. 
Together with preceding results, this indicates \sysname’s generalizability across copyrighted content, dangerous knowledge, and personal information.

\begin{table}[!t]
\caption{The unlearning performance on WMDP-Augment using a domain-adjacent but benign $\mathcal{D}_r$ relative to $\mathcal{D}_f$.
}
\label{tab: comparable forgetting of rmu and kuda}
    \setlength{\tabcolsep}{0.35em}
    \renewcommand{\arraystretch}{1.0}
    \footnotesize
    \centering
\begin{tabular}{cccccc}
\toprule[1pt]
\multicolumn{1}{c}{\multirow{2}{*}{\textbf{Method}}} & \multicolumn{4}{c}{\textbf{Forgetting Quality}}                                      & \textbf{Retaining Utility} \\ \cmidrule(lr){2-5} \cmidrule(l){6-6} 
\multicolumn{1}{c}{}                        & Acc-Bio$\downarrow$         & Acc-Cyber$\downarrow$             & Acc-Avg$\downarrow$              & \textbf{KRD}$\boldsymbol{\rightarrow{} 1}$ & \textbf{MMLU}$\boldsymbol{\uparrow}$                     \\ \midrule
\multicolumn{1}{c}{RMU}                     & 33.22           & 26.57                 & 29.17                 & \cellcolor{blue!24}0.8475                    & \cellcolor{pink!20}51.85             \\
\multicolumn{1}{c}{\textbf{\sysname}}                 & 32.00 & 26.50        & 28.70       & \cellcolor{blue!28}0.8670           & \cellcolor{red!20}54.50           \\
\bottomrule[1pt]
\end{tabular}
\end{table}

\subsection{Mechanism Analysis}
\label{sec:mechanism analysis}
We aim to investigate the underlying reasons why \sysname achieves a better balance between knowledge removal and utility retention than other unlearning baselines. 
Therefore, we conduct a mechanistic analysis experiment from the perspective of \textit{gradient dynamics}. 

\mypara{Setup}
Previous studies~\cite{wang2025gru,wangrethinking} have identified the optimization challenge in gradient-based unlearning methods: the gradients of $\mathcal L_f$ and $\mathcal L_r$ mutually impede each other during the parameter updates.
This highlights the essence of non-conflicting gradients for successful unlearning.
Based on the insight, we design an experiment focusing on gradient conflicts to elucidate the unlearning mechanism of \sysname.

We conduct experiments on WMDP and replace the original $\mathcal D_r$, Wikitext.
Because Wikitext is independent of cyber- or bio-related hazardous content in $\mathcal{D}_f$, it avoids significant conflicts between forgetting and retention objectives, thereby facilitating a high-quality unlearning balance.
To magnify the optimization conflicts, we construct a stringent scenario by replacing $\mathcal{D}_r$ with benign knowledge that is semantically adjacent to $\mathcal{D}_f$.
We compare with RMU, the SOTA method in WMDP, and tune hyperparameters until the RMU-unlearned and \sysname-unlearned models achieve comparable forgetting quality.
To prevent excessive forgetting, we ensure that both models achieve scores above 25 on the Bio and Cyber evaluation tests.
Throughout training, we record the cosine similarity between the gradients of $\mathcal L_f$ and $\mathcal L_r$, and convert it to the angular value in degrees for intuitive interpretability.

\begin{figure*}[!t]
    \centering
    \includegraphics[width=1.0\hsize]{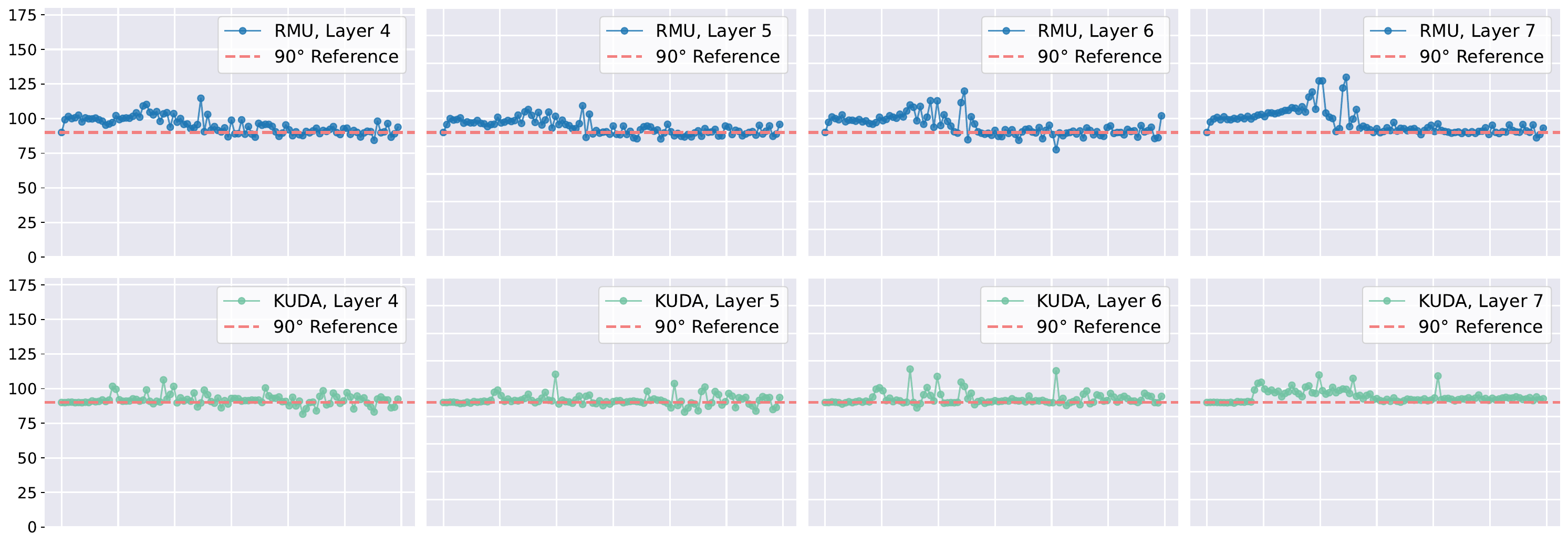}
    \caption{
    Evolution of angles between two gradients during unlearning with RMU and \sysname. 
    We trained on WMDP for 600 steps, showcasing the gradient angle fluctuations throughout the entire process.}
    \label{fig:cos angle}
\end{figure*}

\mypara{Observations}
As shown in \autoref{tab: comparable forgetting of rmu and kuda}, under comparable forgetting quality, RMU incurs more severe degradation of model utility.
This indicates that \sysname allows more precise and targeted knowledge removal, even when it is removed from semantically proximate or related knowledge domains.
\autoref{fig:cos angle} shows the evolution of gradient angles, further revealing the underlying mechanism behind the above observations from the perspective of gradient alignment. 

The gradient angles of RMU predominantly fall within the range of $90^{\circ}\sim 120^{\circ}$, particularly in the early and middle phases.
This suggests a severe imbalance between conflicting objectives, increasing the risk of collapse in model utility.
This may stem from the design of RMU's losses (see introduction in \autoref{subsec: existing solutions}), which employs symmetric MSE for both forgetting and retention, and is thus prone to gradient conflicts.

In contrast, \sysname presents more stable gradient dynamics, with the gradient angles between $\mathcal L_f$ and $\mathcal L_r$ fluctuating around $90^{\circ}$.
This illustrates that \sysname, as a gradient-based method, inherently maintains near-orthogonal alignment between forgetting and retaining gradients during the unlearning process. 
Therefore, it avoids the dilemma of optimization conflicts posed by antagonistic interference, thus eliminating the need to compromise the model utility for knowledge removal. 
We attribute this advantage to the synergy effect of well-designed unlearning losses and the relaxed null-space projection.

\subsection{Hyperparameter Analysis}
\label{subsec: hyperparameter}
We further evaluate the impact of two critical hyperparameters and support our guidance for hyperparameter search.

\mypara{Setup}
We adopt WMDP for hyperparameter analysis due to the well-curated data scale and evaluation efficiency, which enable rapid iteration over variant configurations.
We vary inverse unlearning temperature $\beta$ over a set of $\{0.1, 1.0, 2.0\}$, sample five values for null-space relaxation threshold $\tau$ from the range $3 \times 10^{-4} \sim 3 \times 10^{-3}$, and then conduct a cross-combination study over these hyperparameters.

\mypara{Impact of Inverse Unlearning Temperature $\beta$}
\label{sec:impact of unlearning temperature reciprocal}
\autoref{fig:hyperparameters} presents a heatmap of the hyperparameter experiment results, where darker shades denote better performance.
It can be observed that the choice of $\beta$ significantly affects the unlearning performance: smaller values of $\beta$ lead to stronger knowledge removal but at the cost of degrading the model utility retention.
Moreover, the impact of $\beta$ on forgetting exhibits a similar tendency across different values of $\tau$.

These results correspond to the design of \sysname, where $\beta$ controls the step size of the forgetting process.
In general, as $\beta$ increases, the divergence degree gradually weakens.
This leads to a gradual reduction in the intensity of knowledge removal, with the unlearned model tending toward a state of ``weak forgetting.'' 
Correspondingly, the retaining utility improves progressively, approaching that of the ``origin'' model.

\begin{figure}[!t]
    \centering
    \includegraphics[width=1.0\hsize]{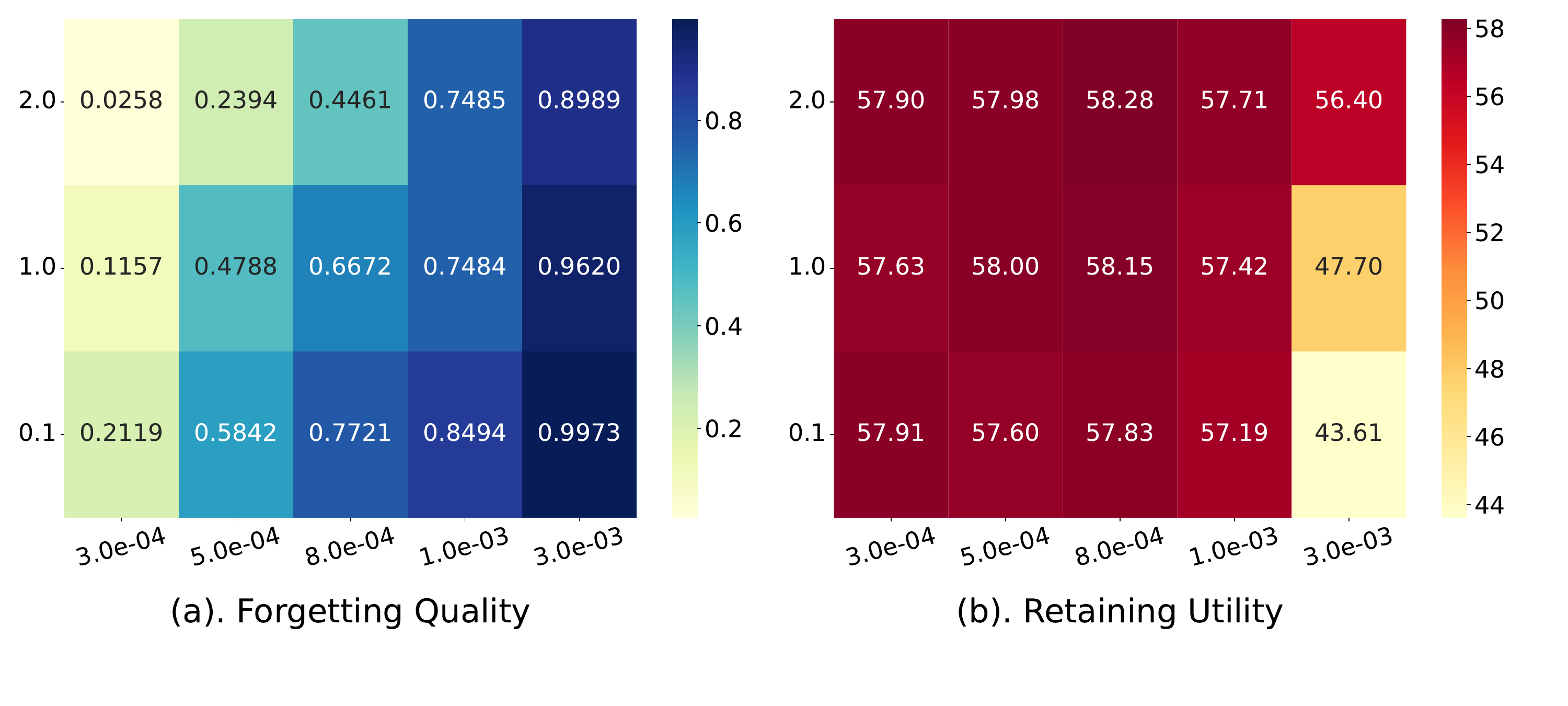}
    \caption{
    Unlearning performance heatmap across hyperparameter combinations of $\beta$ (vertical axis) and $\tau$ (horizontal axis). 
    Plots (a) and (b) are quantified with KRD and MMLU score, respectively.
    Darker colors mean better performance.
    }
    \label{fig:hyperparameters}
\end{figure}

\mypara{Impact of Null-Space Relaxation Threshold $\tau$}
\label{sec:impact of null-Space construction threshold}
As shown in \autoref{fig:hyperparameters}, decreasing $\tau$ significantly enhances retaining utility while gradually weakening forgetting effectiveness.
Notably, the retaining utility is sensitive to variations in $\beta$ under a large $\tau$; however, this sensitivity declines and shifts towards robustness to variations in $\beta$ once $\tau$ drops below a certain value.
We identify the specific $\tau$ as ``\textit{stability boundary}'', at or below which the forgetting process is significantly suppressed, as discussed in \autoref{subsec: principled hyperparameter configuration}.

The observed trend aligns with our predictions, confirming that $\tau$ influences the trade-off between forgetting and retention. 
Specifically, $\tau$ dictates the singular value filtering range, thus controlling the selection of singular vectors used to construct the approximate null-space. 
Lower values of $\tau$ mean fewer singular vectors, which typically correspond to feature directions with minimal contribution to the representation space. 
Therefore, the null-space becomes weakly correlated with the representations of retaining knowledge.
Due to the entanglement discussed in \autoref{subsec: null-space projection}, projecting gradients into this subspace limits perturbation to retaining knowledge, while it also constrains knowledge removal. 
The interaction explains why $\tau$ affects \sysname's sensitivity to changes of $\beta$.

\mypara{Summary}
Based on the comprehensive analysis, we identify the observations:
(1) Decreasing $\beta$ strengthens the forgetting intensity and leads to more aggressive removal of target knowledge;
(2) Decreasing $\tau$ improves retaining utility at the cost of reduced forgetting effectiveness;
(3) Knowledge removal is suppressed if $\tau$ is at or below the stability boundary.
(4) When $\tau$ falls below stability boundary, the variations in $\beta$ yield limited influence on retention.

These observations demonstrate the coupled effects of $\beta$ and $\tau$, where both hyperparameters concurrently modulate the trade-off between knowledge removal and utility retention. 
Beyond their direct influence on unlearning, $\tau$ further modulates how sensitive \sysname is to the variations in $\beta$.
This interdependence highlights the necessity of our two-stage tuning strategy that simplifies the complex joint optimization into two decoupled linear search processes.
Moreover, Observation (3) refines the second stage of our strategy, in which $\tau$ is set to a value slightly larger than the stability boundary.

\subsection{Ablation Study}
\label{sec:ablation study}
We also conduct ablation experiments to verify the effectiveness of \sysname's components.

\mypara{Setup}
Rather than removing each component, we use functional alternatives to conduct ablations.
For \textit{layer selection}, direct removal would make the subsequent stages lose their execution targets: $\mathcal{L}_{\mathrm{KUDA}}$ is obtained from the representation at the last selected layer, while gradients are only calculated for the selected layers. 
We therefore replace the selected layers with five alternative layer groups based on two naive selection criteria: layer depth and component functionality.
(1) Depth-based layers, comprising Shallow ($l_3 \sim l_6$), Middle ($l_{12} \sim l_{15}$), and Deep ($l_{21} \sim l_{24}$) layers; and (2) Functional partitions, focusing on the FFN-dominant ($l_3 \sim l_9$) and MHSA-dominant ($l_{14} \sim l_{20}$) zones as shown in \autoref{fig:CE_zephyr}. 
The layers selected by \sysname are $l_6 \sim l_9$.

For \textit{unlearning loss}, removing $\mathcal{L}_{\mathrm{KUDA}}$ would eliminate the signal that drives target knowledge removal.
We therefore replace it with $\mathcal{L}_{\mathrm{RMU}}$, a representative alternative that likewise provides a representation-level unlearning objective.

For \textit{gradient projection}, training could still proceed after its removal, but the explicit mechanism for balancing forgetting and retention would be lost.
We therefore replace it with GradDiff, a commonly used alternative that performs such balancing through scalar loss weighting $\mathcal L_u = \mathcal L_f + \alpha \cdot \mathcal L_r$, where $\alpha$ is the scaling factor.

\mypara{Effectiveness of Layer Selection}
\label{subsec:effectiveness of layer selection}
\begin{figure}[!t]
    \centering
    \includegraphics[width=1.0\hsize]{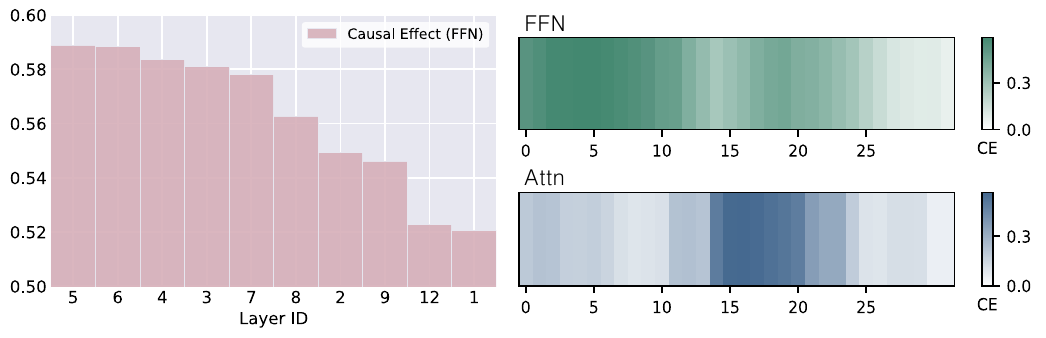}
    \\[-1ex]
    \caption{
    Visualization of Causal Effect (CE).
    The left plot presents the layers with the top-10 CE in FFN components. 
    The right panels present the global causal effects, with darker colors indicating stronger CE. 
    }
    \label{fig:CE_zephyr}
\end{figure}
\begin{figure}[!t]
    \centering
    \includegraphics[width=1.0\hsize]{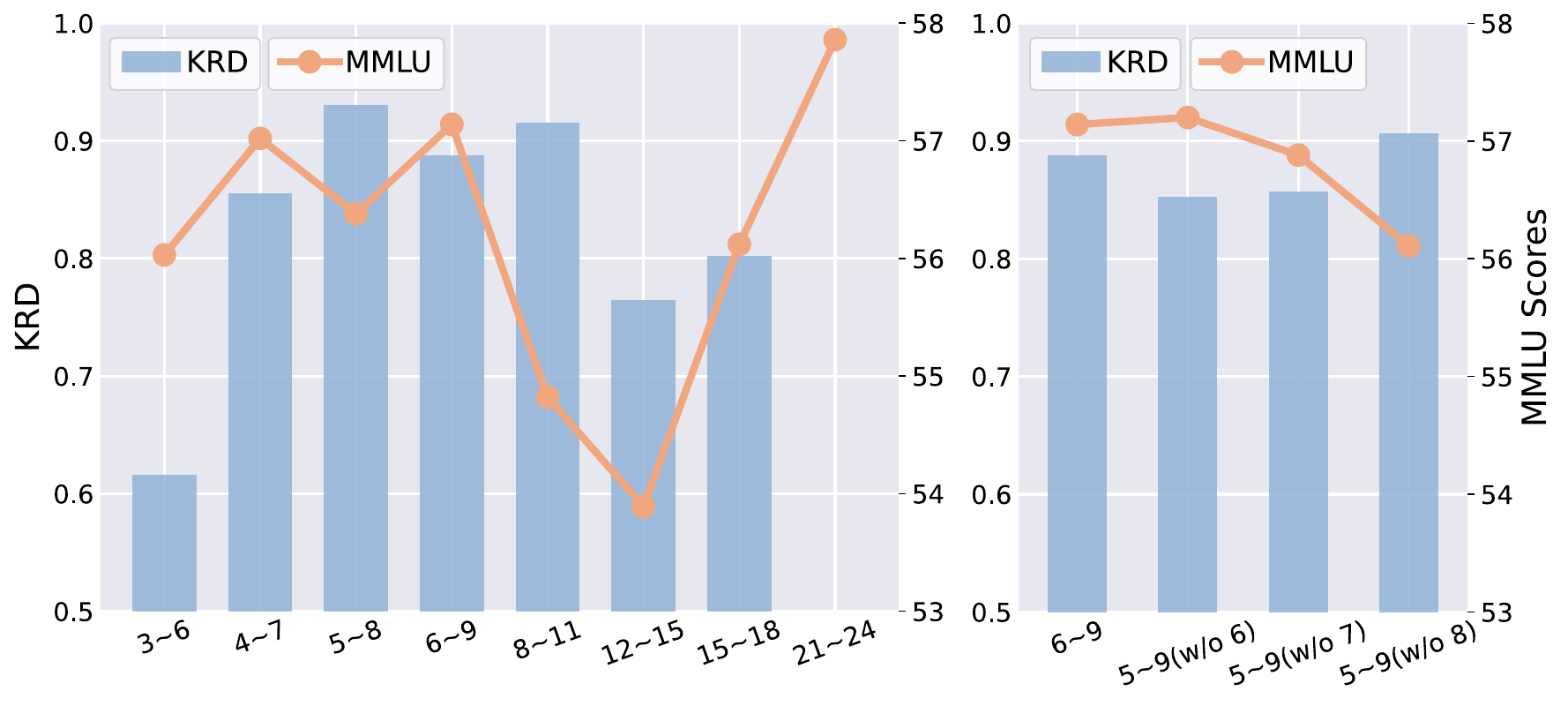}
    \caption{
    Unlearning performance across different unlearning layers. 
    The horizontal axis shows the unlearning layers.
    $i \sim j$ denotes that all layers from $i$ to $j$ are designated as unlearning layers, and ``w/o $k$'' means the excluded layer $k$.
     }
    \label{fig:layer selection}
\end{figure}
\autoref{fig:layer selection} illustrates that the position of unlearning layers significantly affects the unlearning trade-off. Depth-based selections, especially middle and deep layers, fail to achieve effective knowledge removal despite aggregating knowledge via residual streams.
Besides, the layers $l_6 \sim l_9$ selected by \sysname fall within the FFN-dominant region, suggesting that our selection strategy decouples the FFN-dominant and MHSA-dominant layers.
The results in \autoref{fig:layer selection} indicate that this decoupling enables knowledge removal without compromising the models' core utility.
Notably, we observe a robust layer selection interval ($l_4\sim l_9$) that consistently yields superior performance, with minor trade-off fluctuations between forgetting and retention. 
All layers within the interval have high $\text{CE}_{FFN}$, and layers identified by our strategy also fall within this interval.

\mypara{Effectiveness of Knowledge Representation Deviation Loss}
\label{subusec: effectiveness of knowledge representation deviation Loss}
To simplify the analysis, we adopt Acc and MMLU scores on WMDP, along with KnowMem in $\mathcal D_f$ and $\mathcal D_r$ on MUSE, as intuitive metrics of forgetting quality and retaining utility.

\autoref{tab:ablation loss and ns} demonstrates that the relaxed null-space projection is better suited to $\mathcal L_{\text{KUDA}}$ than to $\mathcal L_{\text{RMU}}$. 
The stronger compatibility enables their combination to effectively improve both knowledge removal and utility retention.
This suggests that the asymmetric design of $\mathcal{L}_{\text{KUDA}}$ serves as a fundamental prerequisite that allows relaxed null-space projection to retain model utility without suppressing knowledge removal.

\mypara{Effectiveness of Null-space Projection}
\label{subsec:effectiveness of null-space constraining projection}
\autoref{tab:ablation loss and ns} demonstrates that, with the losses of \sysname, null-space projection outperforms GradDiff in balancing forgetting and retention, consistently across different datasets and models.  
Particularly, at comparable forgetting quality, null-space projection significantly improves retention performance.
We attribute this to its ability to alleviate optimization conflicts through subspace transformations, rather than merely adjusting the relative magnitudes of gradients with a scaling factor.
Therefore, it eliminates the adverse influence of the gradients on the representation space of retaining knowledge.

\begin{table}[!t]
\caption{
The forgetting quality and retaining utility of unlearned models with different loss functions $\mathcal{L}_u$ (\ie, $\mathcal L_{\text{RMU}}$ and $\mathcal L_{\text{KUDA}}$), GradDiff (GD), and Null-Space Projection (NSP).
WMDP-Augment uses benign knowledge that is semantically related to $\mathcal D_f$ as $\mathcal D_r$.}
\label{tab:ablation loss and ns}
    \setlength{\tabcolsep}{0.2em}
    \renewcommand{\arraystretch}{1.0}
    \footnotesize
    \centering
\begin{tabular}{cccc}
\toprule
\textbf{Dataset}
& \textbf{Method} & \textbf{Forgetting Quality $\downarrow$} & \textbf{Retaining Utility $\uparrow$} \\
\midrule
\multirow{4.5}{*}{\textbf{WMDP-Augment}} 
    & $\mathcal L_{\text{RMU}}$ + GD & 28.17 & 51.85 \\
    & $\mathcal L_{\text{RMU}}$ + NSP & 27.45 & 51.22 \\
    \cmidrule(l){2-4}
    & $\mathcal L_{\text{KUDA}}$ + GD & 28.60 & 52.48 \\
    & $ \mathcal L_{\text{KUDA}}$ + NSP & 27.40 & 55.50 \\
\midrule
\multirow{4.5}{*}{\textbf{MUSE-BOOKS}} 
    & $\mathcal L_{\text{RMU}}$ + GD & 30.92 & 58.37 \\
    & $\mathcal L_{\text{RMU}}$ + NSP & 34.58 & 53.60 \\
    \cmidrule(l){2-4}
    & $\mathcal L_{\text{KUDA}}$ + GD & 31.67 & 55.95 \\
    & $ \mathcal L_{\text{KUDA}}$ + NSP & 32.48 & 63.11 \\
\bottomrule
\end{tabular}
\end{table}

\subsection{Robustness Analysis}
\label{subsec: Robustness}
\autoref{subsec: overall performance} presents the effectiveness of \sysname in non-adversarial scenarios.
In practice, adversaries might attempt to recover the removed knowledge.
We illustrate the robustness of \sysname by evaluating the performance against three attacks.

\mypara{Setup}
We launch three adversarial attacks, Min-K$\%$ Prob MIA~\cite{shi2024detecting}, FOCUSONKEY~\cite{zhang2025focusonkey}, and Quantization-based Knowledge Recovery~\cite{zhang2025quantization} on the unlearned models.
The first attack is an MIA specialized for LLMs to probe training data membership leakage. 
The second is a prompt-based attack designed to reproduce removed knowledge from unlearned models.
The third is a quantization-based attack that attempts to recover the knowledge through low-bit precision.
By assessing the unlearned models' resistance to these attacks, we aim to verify whether the supposedly removed knowledge remains latently persistent within the model parameters.
\autoref{fig:Attack} presents the results against the first two attacks, where the \textit{low values} indicate robust knowledge removal. 
Details on the implementation and evaluation are provided in Appendix \ref{sec: Adversarial Attacks}. 

\mypara{Robustness Against MIA}
\autoref{fig:Attack}a presents the results under MIA. 
Ideally, the relative AUC of an unlearned model approaches that of the ``Retrain'' baseline.
While RMU achieves relatively competitive results on BOOKS, it performs poorly on the NEWS dataset, which suggests inconsistent removal degree across datasets.
In contrast, \sysname consistently achieves stable scores close to the ``Retrain'' model on both datasets, indicating superior thoroughness in removing data traces.
The minor gap compared to the ideal results can be attributed to \sysname's objective of preserving utility on knowledge semantically related to $\mathcal{D}_f$, which introduces an inevitable trade-off.
This explanation is supported by the great retention performance of \sysname reported in \autoref{subsec: overall performance}. 

\mypara{Robustness Against Prompt Attack}
We assess defensive robustness against the knowledge recovery attack, with the results shown in \autoref{fig:Attack}b.
Accuracy reflects the ability to answer questions related to the target knowledge.
For RMU, a significant gap can be observed between ``w/o Attack'' and ``w/ Attack.''
This means FOCUSONKEY successfully recovers knowledge, causing a sharp increase in accuracy that even surpasses that of ``Retrain.''
Conversely, \sysname exhibits remarkable robustness. 
Even under attack, it maintains accuracy below that of the ``Retrain'' baseline, especially on NEWS. 
These results indicate that \sysname effectively resists knowledge recovery, substantially limiting the reconstruction of removed knowledge even under specialized attacks.

\begin{figure}[!t]
    \centering
    \includegraphics[width=1\hsize]{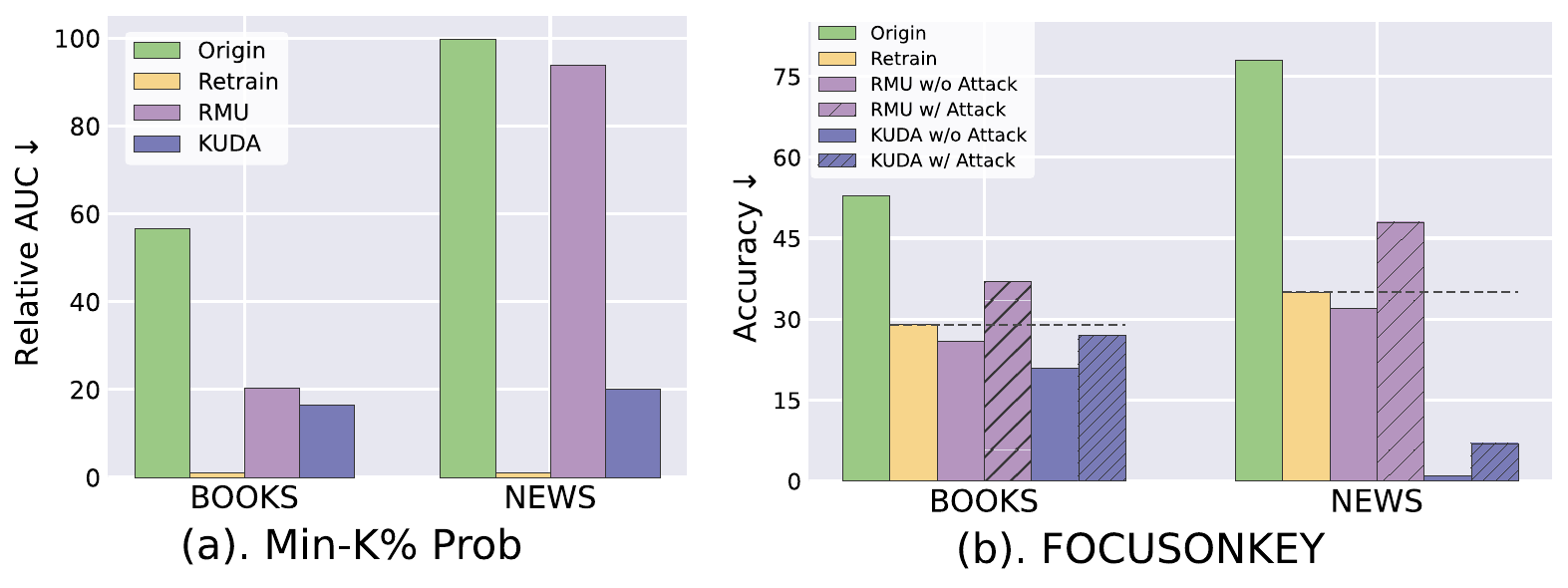}
    \caption{Unlearning performance against adversarial attacks. 
    We select the MUSE benchmark, including BOOKS and NEWS.
    Lower values indicate a more thorough forgetting effect and enhanced defensive robustness.
    }
    \label{fig:Attack}
\end{figure}

\mypara{Robustness Against Quantization}
We apply 4-bit and 8-bit quantization to the models unlearned in BF16 following Zhang \etal~\cite{zhang2025quantization}.
The results show that \sysname remains robust under low-bit quantization, with no evident knowledge recovery.
Due to space limitations, please refer to Appendix~\ref{subsec: quantization} for detailed experiment settings, results, and analysis.

\begin{table}[!t]
\centering
\footnotesize 
\setlength{\tabcolsep}{2pt} 
\renewcommand{\arraystretch}{0.8}
\caption{
The detailed overhead of \sysname in \texttt{HH:MM:SS} format, with 4 unlearning layers.
The cost of parameter update is 1 epoch, and the update without projection is also reported.
}
\label{tab:params}
\begin{tabular}{ccccccc}
\toprule
\textbf{Dataset} & \textbf{Model} & \makecell{\textbf{Causal}\\\textbf{Tracing}} & \makecell{\textbf{Input}\\\textbf{Feature}} & \textbf{SVD} & \textbf{Update} & \makecell{\textbf{Update w/o}\\\textbf{Projection}} \\ \midrule
    \rowcolor[RGB]{227,234,241}
    \cellcolor{white}
\multirow{4.2}{*}{WMDP} 
    & Zephry-Beta     & 2:10:23     & 0:02:40    & 0:09:40 & 0:05:57      & 0:04:28 \\ \cmidrule{2-7} 
    & Llama-3.1  & 2:57:00 & 0:04:16 & 0:09:33          & 0:08:21 & 0:06:25 \\ \cmidrule{2-7} 
    \rowcolor[RGB]{227,234,241}
    \cellcolor{white}
    & Qwen-3      & 2:41:45 & 0:02:32 & 0:18:32          & 0:08:17 & 0:06:34 \\ \midrule
\multirow{2.2}{*}{MUSE} & ICLM (BOOKS)   & 2:21:06 & 0:00:32 & 0:04:08 & 0:12:23  & 0:11:18 \\ \cmidrule{2-7} 
    \rowcolor[RGB]{227,234,241}
    \cellcolor{white}
    & Llama-2 (NEWS) & 1:50:24 & 0:03:04 & 0:04:12 & 0:09:05 & 0:08:25 \\ \bottomrule
\label{tab:computational cost}
\end{tabular}
\end{table}

\subsection{Efficiency Analysis}
\label{subsec: efficiency}
The practicality of \sysname also depends on the computational efficiency.
In this section, we analyze its computational overhead and evaluate the proposed sampling strategy for reducing the cost of input feature capture on large-scale datasets.

\mypara{Computational Overhead}
\autoref{tab:computational cost} presents the computational overhead of \sysname, identifying \textit{causal tracing} as the primary bottleneck. 
However, as discussed in \autoref{subsec: Layers selection}, this is a \textit{one-time} offline operation for each model, and the generated layer-selection results are reusable across datasets. 
The cost of SVD is primarily determined by the dimensionality of the input features, but remains acceptable when restricted to a few unlearning layers.
Moreover, the additional overhead introduced by the projection during parameter updates is marginal across models and unlearning tasks.
We further demonstrate efficiency comparable to that of the baselines, as reported in Appendix~\ref{subsec: efficiency comparison}.

In general, \sysname introduces acceptable additional overhead, which remains stable and does not scale significantly with the size or complexity of the unlearning tasks.
However, the cost of ``Input Feature Capture'' scales linearly with the size of $\mathcal D_r$, leaving an issue we address specifically in the following.

\begin{figure}[!t]
    \centering
    \includegraphics[width=1\hsize]{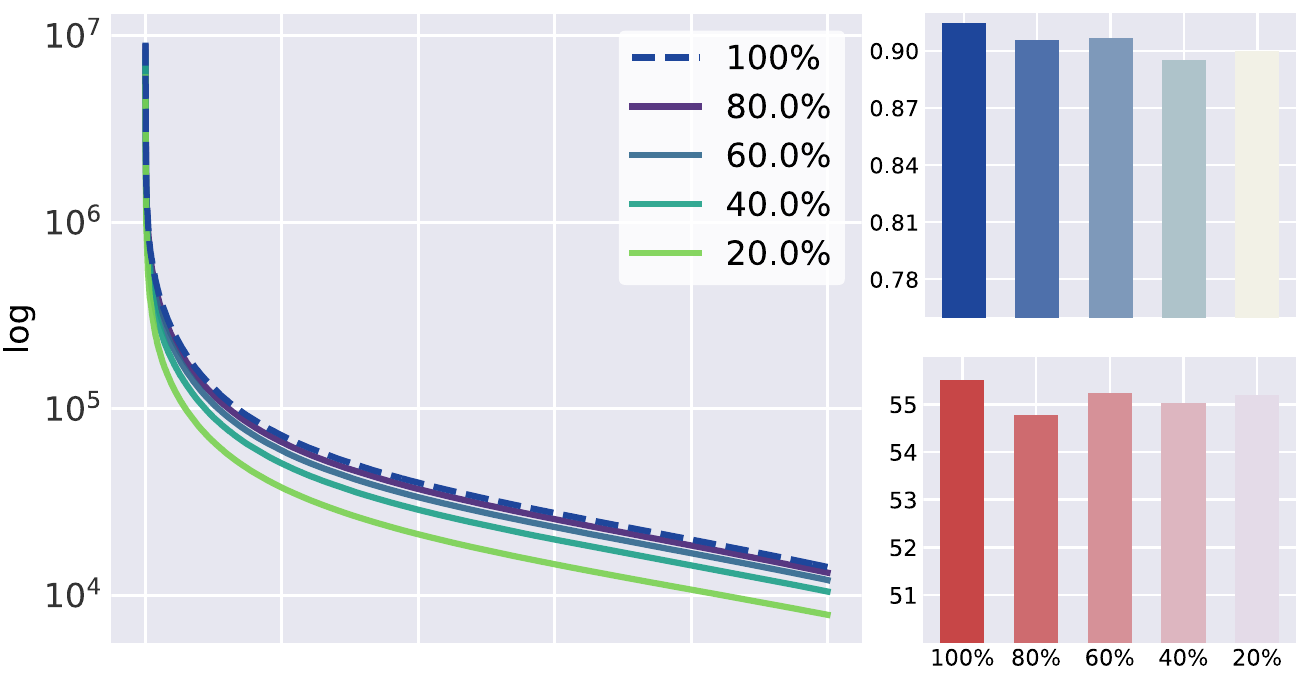}
    \caption{Impact of sampling ratio. 
    The left plot shows eigendecomposition, in descending order on a logarithmic scale.
    The right parts present forgetting quality (top) and retaining utility (bottom), under different sampling ratios and tuning $\tau$.
    }
    \label{fig:efficiency sample}
\end{figure}

\mypara{Sampling Strategy}
To mitigate the computational bottleneck caused by ``Input Feature Capture'' under a large-scale dataset, we propose \textit{stochastic proportional sampling} in \autoref{subsec: null-space projection} to construct an approximate input feature space using a sampled subset of retaining data $\hat{\mathcal{D}}_r$.
A critical consideration is whether such a sampling strategy preserves the information integrity of input features, or if it leads to a degradation in the effectiveness of null-space due to potential information loss.

We conduct a comprehensive analysis with different sampling ratios $p$ and a large-scale corpus with 61,000 samples, following basic unlearning evaluations and spectral decomposition analysis on the input feature spaces.
As presented in \autoref{fig:efficiency sample}, the high consistency of eigenvalue spectra confirms that stochastic sampling captures the primary information and core distribution of the original representation space.
Since relaxed null-space projection primarily targets preserving these dominant features, the inevitable minor information loss only induces variations in the construction of the null-space, without compromising its protection function.
The right plots of \autoref{fig:efficiency sample} confirm that near-optimal unlearning results can be achieved across ratios $p$ through minor adjustments of $\tau$, while reducing the overhead to $\frac{p}{100}$ of the original one, thus balancing computational efficiency with the unlearning effectiveness.

\section{Related Work}
\subsection{LLM Unlearning}
\label{subsec: LLM Unlearning Implementation}
LLM unlearning is an extension of machine unlearning~\cite{chen2021machine,chen2022graph,machineunlearning, chen2024machine} that is specifically applied to LLMs, to remove undesired knowledge from a pre-trained LLM without full retraining~\cite{liu2025rethinking, ren2025sokmachineunlearninglarge}. 
Representative methods for LLM unlearning primarily rely on parameter fine-tuning. 
This includes full parameter updates based on specific forgetting losses~\cite{jang-etal-2023-knowledge,zhangnegative,fan2024simplicity}, which often incorporate retaining losses to balance forgetting and retention~\cite{yao2024large}, serving as the foundational approach. 
To pursue higher efficiency and balance, some studies focus on identifying parameters correlated to target knowledge~\cite{wu-etal-2023-depn,WALGE, pochinkov2024dissectinglanguagemodelsmachine}. 
However, these are computationally expensive, restricting practical applicability.
Besides, researchers explore partial parameter updates based on activation steering~\cite{li2024wmdp,shen2025lunar}, but the effectiveness of forgetting is limited~\cite{fan2025llmunlearningresilientrelearning,yang2025faithun}.
More related works are deferred to Appendix \ref{sec: More Related Works of LLM Unlearning}.

\subsection{Model Editing}
\label{subsec: Model Editing}
Model editing aims to modify LLMs' knowledge by redirecting inputs to specific outputs, often for factual associations that can be formulated as triples $\langle subject, relation, object \rangle$, where the task is to replace the original ``object'' with a given new target~\cite{modeleditingsurvey,editingsurvey}.
Among existing approaches, ``Locate-and-Edit'' has emerged as a mainstream paradigm~\cite{modeleditingsurvey}.
The representative method includes ROME~\cite{meng2022locating}, MEMIT~\cite{mengmass}, EMMET~\cite{gupta2024emmet}, and AlphaEdit~\cite{fangalphaedit}.
Despite the similar technical intuition, \sysname is distinct from model editing.
Model editing redirects knowledge mappings to the given outputs, while LLM unlearning focuses on removing such mappings. 
This lack of substitute targets makes LLM unlearning confront severe utility degradation.
Besides, model editing typically focuses on facts with relative locality. 
LLM unlearning, however, addresses broader concepts or data traces leading to knowledge entanglement. 
Moreover, model editing operates on structured factual triples, while LLM unlearning may target paragraphs without such structure. 
These gaps hinder the direct transfer of model editing algorithms to LLM unlearning scenarios.

\section{Conclusion}
In this paper, we focus on the LLMs' property of knowledge storage for effective unlearning. 
We propose \sysname, a representation-based method that enables targeted parameter updates for knowledge removal with minimal degradation on model utility.
\sysname utilizes a component-level identification method with a sliding window strategy to select unlearning layers, introduces a representation deviation mechanism to remove target knowledge, and employs the relaxed null-space projection to preserve the representations of retaining knowledge.
Besides, we propose a two-stage tuning strategy to decouple the joint hyperparameter searches.
Experimental results demonstrate that \sysname outperforms most representative baselines in balancing both knowledge removal and utility retention, and we further reveal the underlying mechanism.

\section*{Acknowledgments}
This project was supported by the National Science and Technology Major Project (No. 2026ZD0128300), the National Natural Science Foundation of China (No. U23A20296, 62441618, 62402431), Zhejiang Province's "Lingyan" R\&D Project (No. 2025C01195), Zhejiang Provincial Natural Science Foundation of China (No. LZ26F020002), Natural Science Foundation of Jiangsu Province (No. BK20220075), and the project CiCS of the research programme Gravitation which is (partly) financed by the Dutch Research Council (NWO) under Grant No. 024.006.037.


\bibliographystyle{plain}
\bibliography{sample-base.bib}

\appendices

\section{More Related Work of LLM Unlearning}
\label{sec: More Related Works of LLM Unlearning}
\subsection{LLM Unlearning via External Intervention}
\label{subsec: LLM Unlearning via External Intervention}
In contrast to parameter modifications, the \textit{external intervention} methods introduce external components to interfere with the generation process, such as context fine-tuning~\cite{shi2023context}, or by adding perturbations to input embeddings~\cite{ECO} or output logits~\cite{ULD}.
Therefore, these methods control the model to exhibit ``unlearning'' behavior without updating the model parameters.
However, these only achieve unlearning at the behavioral level, failing to remove the target knowledge encoded within the model's parameters.
This renders such methods vulnerable to the risk of target knowledge being extracted from the model’s intermediate representations~\cite{dong2025interninverse}.
These approaches are excluded from our discussion and evaluation because they do not fulfill the fundamental requirement of LLM unlearning, specifically the removal of targeted knowledge. 
Therefore, they are considered distinct from the scope of our research.

\subsection{Attacks against LLM Unlearning}
\label{subsec: LLM Unlearning Attacks}
The robustness of an unlearned LLM against adversarial attacks is also a critical aspect for ensuring practical applicability~\cite{liu2025rethinking}.
On one hand, existing attacks aim to extract forgotten knowledge by exploiting residual memory through adversarial attacks~\cite{zhang2025focusonkey} or backdoor triggers~\cite{huutien2025improvingllmunlearningrobustness}.
On the other hand, the adversary recovers the forgotten knowledge through specific prompts~\cite{zhang2025focusonkey} that simply emphasize certain keywords in the prompts. 
Notably, even without access to the forgetting set, the supposedly removed information can be restored by fine-tuning on retaining sets~\cite{hu2024jogging} or quantizing the models to low precision~\cite {zhang2025quantization}.
These reveal the vulnerability of existing knowledge-removal-intended methods, indicating that they tend to obscure knowledge rather than remove it~\cite{ren2025sokmachineunlearninglarge}. 

\subsection{Application of LLM Unlearning}
\label{subsec: LLM Unlearning Application}
Explorations across diverse application domains have revealed a broad potential of LLM unlearning, including the removal of privacy-sensitive training data encoded in LLMs~\cite{jang-etal-2023-knowledge, wu-etal-2023-depn,codeunlearning}, the suppression of harmful content generation~\cite{li2024wmdp}, and the enforcement of copyright compliance~\cite{DouLLDW25copyright}.
Furthermore, related research has extended its scope to address training data contamination, encompassing issues such as noisy data~\cite{wang2025llm}, mislabeled content~\cite{shilov2025beyond}, and backdoor attacks~\cite{jiang2025backdoor}.

\section{Proof}
\label{sec: proof}
We first restate the definition of null-space as follows:
\begin{definition}[Null-space]
\label{def: null-space}
For a matrix $A\in \mathbb{R}^{m\times n}$, its null-space is defined as
\begin{equation}
\mathrm{Null}(A)=\{v\in\mathbb{R}^{n}\mid Av=0\}
\nonumber
\end{equation}
That is, the null-space $\mathrm{Null}(A)$ contains all vectors that are mapped to zero by $A$.
\end{definition}

For the model layer $l$, let $x^l \in \mathbb{R}^{1 \times d_1}$ denote the input feature of the sample $x$ at layer $l$, \ie, the output of layer $l-1$ and the input to layer $l$. 
In our setting, the input feature space of $x^l$ is defined as
\begin{equation}
    X^l=(x^l)^\top x^l \in \mathbb{R}^{{d_1}\times {d_1}}.
    \nonumber
\end{equation}

Let $\Delta w^l \in \mathbb{R}^{d_1 \times d_2}$ denote the parameter update at layer $l$.
We write it column-wise as
\begin{equation}
\Delta w^l = [\delta_1,\delta_2,\dots,\delta_{d_2}], \qquad \delta_j \in \mathbb{R}^{d_1}.
\nonumber
\end{equation}
When we say that \textit{$\Delta w^l$ lies in the null space of $X^l$}, this means that each column of $\Delta w^l$ belongs to $\mathrm{Null}(X^l)$, namely
\begin{equation}
\delta_j \in \mathrm{Null}(X^l), \qquad \forall j=1,\dots,d_2.
\nonumber
\end{equation}

\begin{theorem}
When the parameter gradient $\Delta w^l$ lies in the null-space of the input feature space $X^l$, there is:
\begin{equation}
    x^l \cdot \Delta w^l = 0
    \nonumber
\end{equation}
\end{theorem}

\begin{proof}
According to \autoref{def: null-space}, this is equivalent to
\begin{equation}
X^l \delta_j = 0, \qquad \forall j=1,\dots,d_2.
\nonumber
\end{equation}
Since $X^l = (x^l)^\top x^l$, we have
\begin{equation}
\begin{aligned}
0=
\delta_j^\top X^l \delta_j
= 
\delta_j^\top (x^l)^\top x^l \delta_j 
=
\|x^l \delta_j\|_2^2 
\nonumber
\end{aligned}
\end{equation}
Therefore, $X^l\delta_j = 0$ implies
\begin{equation}
x^l \delta_j = 0.
\nonumber
\end{equation}
Stacking the columns back together, we obtain
\begin{equation}
\begin{aligned}
x^l \cdot \Delta w^l
&=
x^l \cdot [\delta_1, \delta_2, \dots, \delta_{d_2}] \\
&=
[x^l  \delta_1, x^l \delta_2, \dots, x^l \delta_{d_2}]\\
&=
0
\nonumber
\end{aligned}
\end{equation}
which is equivalent to
\begin{equation}
x^l \Delta w^l = 0
\nonumber
\end{equation}
Hence, the parameter update does not alter the layer-$l$ output on previous-task features.

\end{proof}

\section{Experimental Details}
\subsection{Evaluation Metrics}
\label{subsec: Metrics}

\mypara{MUSE~\cite{shi2024muse}}
It mainly leverages the text-generation capacity of LLM and designs three complementary metrics to holistically quantify the unlearning performance.
\begin{itemize}
    \item \textbf{Verbatim Memorization (VerbMem)}
        This metric quantifies the verbatim reproduction of contexts from data in the forgetting set. 
        We evaluate this by inputting the LLM with the initial tokens $x_{<t}$ of a text sequence $x \in \mathcal D_f$, and comparing the similarity between the generated continuation and the ground-truth one $x_{[t:end]}$:
        \begin{equation}
            \mathrm{VerbMem}(\mathcal{M},\mathcal D_f) = \frac{1}{n} \sum_{x \in \mathcal D_f} \mathrm{ROUGE}(\mathcal{M}(x_{<t}),x_{[t:end]}),
            \nonumber
        \end{equation}
        where $n$ is the sample size of $\mathcal D_f$, $\mathrm{ROUGE}$ denotes the ROUGE-L F1 score.

    \item \textbf{Knowledge Memorization (KnowMem)}
        This metric evaluates the model's ability to utilize target knowledge across various generation tasks through the format of knowledge-based question-answering. 
        The target knowledge can originate from either the forgetting set or the retaining set.
        Specifically, prompt the LLM with questions $Q$, and compare the similarity between its output and the ground-truth answer $A$:
        \begin{equation}
            \mathrm{KnowMem}(\mathcal{M},\mathcal{D}) = \frac{1}{|\mathcal{D}|} \sum_{(Q,A)) \in \mathcal{D}} \mathrm{ROUGE}(\mathcal{M}(Q),A),
            \nonumber
        \end{equation}
        where $\mathcal{D}$ can be $\mathcal D_f$ or $\mathcal D_r$. 
        Notably, $(Q,A)$ pairs are specifically designed based on text excerpts of $\mathcal{D}$, thereby reflecting the LLM's accurate comprehension of the questions and its ability to appropriately utilize knowledge. 

    \item \textbf{Privacy Leakage (PrivLeak)}
        This metric quantifies the privacy leakage degree of the unlearned model by measuring its defense to Membership Inference Attacks (MIA), a common privacy attack. 
        It assesses whether the model inadvertently reveals the membership that $x \in \mathcal D_f$ was used for training. 
        Specifically, apply a Min-K\% Prob~\cite{shi2024detecting} MIA method to the model and compute the corresponding AUC-ROC score of discriminating members datasets and non-members datasets.
        Based on this, PrivLeak is defined by comparing with the ``retrain'' model:
        \begin{equation}
            \mathrm{PrivLeak} = \left| \frac{
                \mathrm{AUC}(\mathcal{M}_{un};\mathcal D_f, \mathcal D_{h}) - \mathrm{AUC}(\mathcal{M}_{re};\mathcal D_f, \mathcal D_{h})
            }{
            \mathrm{AUC}(\mathcal{M}_{re};\mathcal D_f, \mathcal D_{h})
            } \right|,
            \nonumber
        \end{equation}
        where $\mathcal D_h$ is the hold-out (non-member) set which the model has never been trained on, $\mathcal M_{un}$ and $\mathcal M_{re}$ denote the unlearned model and the retrained model, respectively.
    
    \item \textbf{Knowledge Removal Degree (KRD)}
    This metric provides a holistic evaluation of forgetting quality: when any forgetting metric reflects poor performance of knowledge removal, KRD approaches 0; 1 indicates complete forgetting.
    The harmonic mean achieves this effectively by averaging the reciprocals.
    Moreover, the disruption caused by excessive forgetting should be eliminated in the computation process. 
    For instance, when a metric yields a value better than that of the ``retrain'' model, its forgetting quality should be considered equivalent to that of the ``retrain'' model, rather than superior.
    For MUSE, we use the metric value of the ``retrain'' model as an upper limit for truncation.
    Formally, KRD is defined as: 
    \begin{gather}
    \begin{aligned}
        \mathcal{FQ}^1_{M} &= \frac{\mathrm{VM}(\mathcal M_o) - \max\{\mathrm{VM}(\mathcal M_f), \mathrm{VM}(\mathcal M_r)\}}{\mathrm{VM}(\mathcal M_o) - \mathrm{VM}(\mathcal M_r)} \\[1pt]
        \mathcal{FQ}^2_{M} &= \frac{\mathrm{KM}(\mathcal M_o) - \max\{\mathrm{KM}(\mathcal M_f), \mathrm{KM}(\mathcal M_r)\}}{\mathrm{KM}(\mathcal M_o) - \mathrm{KM}(\mathcal M_r)} \\[1pt]
        \mathcal{FQ}^3_{M} &= \frac{ \mathrm{PL}(\mathcal M_o) - \max\{\mathrm{PL}(\mathcal M_f), \mathrm{PL}(\mathcal M_r) \}}{\mathrm{PL}(\mathcal M_o) - \mathrm{PL}(\mathcal M_r)}
        \nonumber
    \end{aligned} \\[1pt]
        \mathrm{KRD}_{\text{MUSE}} = \frac{3}{\sum_{i=1}^{3} \frac{1}{\mathcal{FQ}^i_{M}}}, \nonumber
    \end{gather}
    where $\mathrm{VM}$, $\mathrm{KM}$ and $\mathrm{PL}$ represent VerbMem, KnowMem on $\mathcal{D}_f$, and PrivLeak, respectively.
    Similarly, $\mathcal M_o$, $\mathcal M_r$ and $\mathcal M_f$ are ``origin'', ``retrain'' and unlearned models.
\end{itemize}

\mypara{WMDP~\cite{li2024wmdp}}
It assesses the LLM's comprehension of specific knowledge through a multiple-choice format.
\begin{itemize}
    \item \textbf{Answer Accuracy (Acc)}
    WMDP evaluates the model's comprehension of hazardous knowledge by measuring its accuracy in answering multiple-choice questions. 
    It comprises a large-scale collection of questions across different domains of harmful information.
    During evaluation, the model is prompted with a question and four candidate options.
    The overall accuracy is then computed based on the model's responses:
    \begin{equation}
            \mathrm{Acc} =\frac{
                \sum \mathds{1}(\mathcal{M}(Q)=A)
            }{N_{total}},
            \nonumber
        \end{equation}
    where $\mathds{1}$ is the indicator function, and $N_{total}$ is the number of correctly answered questions and total questions.

    \item \textbf{Knowledge Removal Degree (KRD)}
    The KRD is consistent with that of MUSE.
    However, due to the absence of the ``retrain'' model (the reason is detailed \autoref{subsec:experimental setup}), we truncate the metric values to 25, which is the desired test score for forgetting, indicating that the model answers randomly. 
    Therefore, KRD for WMDP is formulated as:
    \begin{gather}
    \begin{aligned}
        \mathcal{FQ}^1_{W} &= \frac{\mathrm{Cyber}(\mathcal M_o) - \max\{\mathrm{Cyber}(\mathcal M_f), 25\}}{\mathrm{Cyber}(\mathcal M_o) - 25} \\[1pt]
        \mathcal{FQ}^2_{W} &= \frac{\mathrm{Bio}(\mathcal M_o) - \max\{\mathrm{Bio}(\mathcal M_f), 25\}}{\mathrm{Bio}(\mathcal M_o) - 25} 
        \nonumber
    \end{aligned} \\[1pt]
        \mathrm{KRD}_{\text{WMDP}} = \frac{2}{\sum_{i=1}^{2} \frac{1}{\mathcal{FQ}^i_{W}}}, \nonumber
    \end{gather}
    where $\mathrm{Cyber}$ and $\mathrm{Bio}$ represent answer accuracy in test sets related to Cyber and Bio, respectively.
    Besides, $\mathcal M_o$ and $\mathcal M_f$ denote ``origin'' and unlearned models.

    \item \textbf{MMLU Scores (MMLU~\cite{hendrycks2020measuring})}
    WMDP adopts MMLU to evaluate a LLM’s general knowledge that should be retained after unlearning.
    MMLU benchmark provides a comprehensive evaluation of a LLM’s knowledge. 
    It comprises 57 subjects spanning a wide range of fields, covering basic mathematics, U.S. history, computer science, etc. 
    With a difficulty spectrum from elementary to advanced, MMLU is suitable for testing models at different levels of proficiency. 
    Similarly, MMLU employs a multiple-choice format consistent with the WMDP test sets, and its scoring metric is based on the overall proportion of questions the model answers correctly across all subjects.
\end{itemize}

\subsection{Baseline Implementation}
\label{subsec: Baselines Implementation}
\mypara{MUSE}
Both GA and NPO are implemented with GradDiff and employ the optimal retaining loss following Shi \etal~\cite{shi2024muse} (KL divergence for BOOKS, and Cross-Entropy for NEWS).
Likewise, SimNPO is paired with the corresponding retaining loss according to Fan \etal~\cite{fan2024simplicity}. 
The implementation of WHP directly uses the code provided by Shi \etal~\cite{shi2024muse}.
For RMU on MUSE, due to the absence of an open-source implementation, we adopt the MUSE codebase to reproduce RMU.

\mypara{WMDP}
We implemented GA, NPO, and SimNPO following the settings of Fan \etal~\cite{fan2024simplicity}.

For all baselines, we adopt the hyperparameter configurations reported in the benchmarks or their respective papers to get the optimal performance.
Regarding the implementation of RMU in the MUSE benchmark, we perform a grid search to identify the optimal hyperparameter configuration.

\begin{table}[!t]
\caption{The detailed hyperparameter configuration of \sysname for implementation across different datasets and models.
Within the 'Epoch' column, the notation $i-j$ means that the model is trained with a total of $i$ epochs, with the $j$-th checkpoint selected as the optimal performance for our evaluation.
Others default to the final model.
}
\label{tab:Implementations}
\setlength{\tabcolsep}{0.4em}
\renewcommand{\arraystretch}{1.2}
\footnotesize
\centering
\begin{tabular}{cccccc}
\toprule[1pt]
\textbf{Dataset} & \textbf{Model} & $\boldsymbol{\beta}$ & $\boldsymbol{\tau}$ & \textbf{Learning Rate} & \textbf{Epoch} \\ \midrule
    \rowcolor[RGB]{227,234,241}
    \cellcolor{white}
\multirow{4.2}{*}{\textbf{WMDP}}
    & Zephry-7B-Beta     & $2.0$ & $2\times10^{-3}$ & $1\times10^{-5}$ & $1$ \\ \cmidrule{2-6}
    & Llama-3.1-8B  & $0.1$ & $3\times10^{-3}$ & $1\times10^{-5}$ & $1$ \\ \cmidrule{2-6}
    \rowcolor[RGB]{227,234,241}
    \cellcolor{white}    
    & Qwen3-8B      & $1.0$ & $4\times10^{-3}$ & $1\times10^{-4}$ & $1$ \\ \midrule
\textbf{MUSE-BOOKS} 
    & ICLM-7B       & $0.1$ & $8 \times 10^{-2}$ & $1 \times 10^{-5}$ & $5$ \\ \midrule
    \rowcolor[RGB]{227,234,241}
    \cellcolor{white}
\textbf{MUSE-NEWS}  
    & Llama-2-7B    & 2.0 & $4 \times 10^{-3}$ & $5 \times 10^{-5}$ & $3-2$ \\ \bottomrule[1pt]
\end{tabular}
\end{table}

\subsection{\sysname Implementation}
\label{subsec: Implementations}
All experiments are conducted on a server equipped with two NVIDIA L40 (48GB) GPUs and 512GB of system RAM. 
During the training phase, we employ a batch size of 4. 
Due to the memory requirement of null-space projection matrices and the constraints of limited GPUs' memory, we implement a layer-wise update mechanism that necessitates frequent matrix transfers between CPU and GPU memory. 
Although this enables execution under hardware constraints, the resulting I/O overhead introduces an extra performance delay. 
Consequently, the operational efficiency of \sysname is scalable, suggesting that the execution speed could be further optimized in environments with sufficient GPU memory to accommodate all matrices locally.

In \autoref{subsec: overall performance}, the process of hyperparameter selection adheres to our proposed two-stage tuning strategy to facilitate efficient configuration.
By transitioning from a purely empirical search to this structured pipeline, we significantly reduce the search space complexity while ensuring that the selected parameters achieve a balance between forgetting and retention.
The detailed hyperparameter configurations are presented in \autoref{tab:Implementations}.
Besides, we set $N=10$ and $s=4$ for layer selection.
Due to the small scale of datasets used in \autoref{subsec: overall performance}, we set the sampling ratio $p=100$.

Notably, it is observed that identical hyperparameter configurations can lead to different unlearning outcomes across different hardware environments, such as GPU architectures. 
However, the fundamental trends of hyperparameter sensitivity remain consistent with our analysis in \autoref{subsec: hyperparameter}.
This suggests that suitable hyperparameters across different hardware devices can be efficiently identified using the two-stage tuning strategy.

\section{Additional Evaluation on TOFU}
\label{sec: tofu}

\subsection{Experimental Setup}
TOFU~\cite{mainitofu} is a benchmark designed to evaluate LLM unlearning on question-answering about fictitious authors' personal information. 
We adopt TOFU as a complementary benchmark to further examine the generalizability of KUDA to knowledge of personal information. 
In this section, we introduce the dataset, models, and evaluation metrics, followed by the experimental results and analysis.

\mypara{Dataset}
TOFU contains 200 author profiles synthesized by GPT-4, with 20 question-answer pairs for each author.
TOFU provides three forgetting settings, where 1\%, 5\%, and 10\% of the authors constitute the forgetting set $\mathcal{D}_f$, and the remaining authors constitute the corresponding retaining set $\mathcal{D}_r$.

Forgetting quality is evaluated on $\mathcal{D}_f$, which contains QA data related to the unlearning targets. 
Retaining utility is evaluated on $\mathcal{D}_r$, and two auxiliary datasets, Real Authors $\mathcal{D}_{RA}$ and World Facts $\mathcal{D}_{WF}$.
Specifically, $\mathcal{D}_r$ measures the preservation of knowledge about the remaining fictitious authors, while $\mathcal{D}_{RA}$ and $\mathcal{D}_{WF}$ assess the retention of knowledge about neighboring real authors and general world facts, respectively.

\mypara{Model}
Since the profiles are entirely synthetic, the corresponding knowledge is absent from the pre-training data.
TOFU therefore constructs the ``origin'' model by fine-tuning a pre-trained LLM on the full author profiles QA dataset, such that the model learns knowledge about all fictitious authors.
The ``retrain'' model is obtained by fine-tuning the same pre-trained LLM on $\mathcal{D}_r$, without exposure to $\mathcal{D}_f$, and thus serves as the golden baseline for forgetting.

\mypara{Metric}
We adopt the comprehensive evaluation framework proposed by Yuan \etal~\cite{yuan2025closer}, which extends the original TOFU evaluation to assess model outputs from multiple dimensions, such as semantic consistency, factual correctness, etc.
\begin{itemize}
    \item \textbf{ROUGE (R).}
    This measures the word-level overlap between the generated response and the ground-truth answer using ROUGE-L recall:
    \begin{equation}
        \mathrm{R}
        =
        \frac{1}{|\mathcal{D}|}
        \sum_{(x,y)\in\mathcal{D}}
        \mathrm{ROUGE\text{-}L}(\mathcal{M}(x),y),
        \nonumber
    \end{equation}
    where $\mathcal{M}(x)$ is the output answer generated by the model $\mathcal{M}$.
    A higher R indicates more lexical overlap with the ground-truth answer.

    \item \textbf{Probability (P).}
    This measures the model’s ability to predict the ground-truth answer:
    \begin{equation}
        \mathrm{P}
        =
        \frac{1}{|\mathcal{D}|}
        \sum_{(x,y)\in\mathcal{D}}
        \frac{1}{T_y}
        \sum_{t=1}^{T_y}
        p_{\mathcal{M}}(y_t\mid x, y_{<t}),
        \nonumber
    \end{equation}
    where $T_y$ is the token length of $y$.
    A higher P indicates greater confidence in the ground-truth answer.
    
    \item \textbf{Truth Ratio (TR).}
    This measures the model's relative preference between a correct answer $\tilde{y}$ and incorrect ones $\hat{y}$:
    \begin{equation}
        r(x;\mathcal{M})
        =
        \frac{1}{|\hat{\mathcal{Y}}|}
        \sum_{\hat{y}\in\hat{\mathcal{Y}}}
        \frac{\mathrm{P}(\hat{y}\mid x;\mathcal{M})}
        {\mathrm{P}(\tilde{y}\mid x;\mathcal{M})}.
        \nonumber
    \end{equation}
    For forgetting, successful unlearning requires the model to be unable to distinguish the correct answer from the incorrect ones, \ie, $r\approx1$. 
    For retention, the model should retain a strong preference for the correct answers across $\mathbb{D}_u=\{\mathcal{D}_r,\mathcal{D}_{\mathrm{RA}},\mathcal{D}_{\mathrm{WF}}\}$, \ie, $r \rightarrow 0$. 
    Therefore, TR is defined as:
    \begin{equation}
        \mathrm{TR}
        =
        \begin{cases}
            1-\min\left\{r,\dfrac{1}{r}\right\},
            & \mathcal{D}=\mathcal{D}_f, \\
            \max\{0,1-r\},
            & \mathcal{D}\in\mathbb{D}_u.
        \end{cases}
        \nonumber
    \end{equation}
    A lower TR indicates better forgetting on $\mathcal{D}_f$, whereas a higher TR indicates better knowledge retention on $\mathbb{D}_u$.

    \item \textbf{Token Entropy (TE).}
    This measures the token diversity and readability of generated responses:
    \begin{equation}
        \mathrm{TE}
        =
        -\frac{\sum_{i=1}^{m} f(w_i)\log_2 f(w_i)}
        {\log_2 |\mathcal{M}(x)|},
        \nonumber
    \end{equation}
    where $m$ is the number of unique tokens and $f(w_i)$ is the normalized frequency of the unique token $w_i$. 
    A higher value indicates less repetitive tokens.

    \item \textbf{Cosine Similarity (CS).}
    This measures the semantic similarity between responses generated before and after unlearning:
    \begin{equation}
        \mathrm{CS}
        =
        \max \left \{
        \cos \big \langle
        E \left(\mathcal{M}_o(x)\right),E\left(\mathcal{M}_{un}(x)\right)
        \big \rangle,0
        \right \},
        \nonumber
    \end{equation}
    where $E(\cdot)$ denotes the Sentence-BERT encoder. A higher value indicates better preservation of the original responses.

    \item \textbf{Entailment Score (ES).}
    This measures whether the generated responses convey the information of the ground truth using a natural language inference (NLI) model:
    \begin{equation}
        \mathrm{ES}
        =
        \frac{1}{|\mathcal{D}|}
        \sum_{(x,y)\in\mathcal{D}}
        \mathds{1}
        \left[
        \mathrm{NLI}\left(\mathcal{M}(x),y\right)
        =
        \text{entailment}
        \right],
        \nonumber
    \end{equation}
    where $\mathds{1}$ is the indicator function. 
    A higher ES indicates better factual consistency with the ground-truth answers.
\end{itemize}

The above metrics evaluate model responses from complementary perspectives, covering both expression and semantic properties. 
To provide a more intuitive summary of forgetting and retention, Yuan \etal further aggregate these multidimensional metrics into two measures: 
\textit{Forget Efficacy (FE)} for forgetting quality, and \textit{Model Utility (MU)} for retaining utility.

\mypara{Forget Efficacy (FE)}
This aggregates these metrics (except $\mathrm{TE}$) on $\mathcal{D}_f$ to measure the overall forgetting quality:
\begin{equation}
    \mathrm{FE}
    =
    1-
    \frac{
    \mathrm{R}_{f}+
    \mathrm{P}_{f}+
    \mathrm{TR}_{f}+
    \mathrm{CS}_{f}+
    \mathrm{ES}_{f}
    }{5}
    \nonumber
\end{equation}
A higher FE indicates better forgetting quality.
When FE of an unlearned model exceeds that of the ``Retrain'' model, it is considered to be the complete removal of target knowledge.

\mypara{Model Utility (MU)}
MU aggregates all six metrics on the utility retention sets $\mathbb{D}_u=\{\mathcal{D}_r,\mathcal{D}_{\mathrm{RA}},\mathcal{D}_{\mathrm{WF}}\}$. 
Let $\mathbb{M}_u=\{\mathrm{R},\mathrm{P},\mathrm{TR},\mathrm{TE},\mathrm{CS},\mathrm{ES}\}$, MU is defined as the harmonic mean:
    \begin{equation}
        \mathrm{MU}
        =
        \operatorname{HMean}
        \left(
        \left\{
        q(\mathcal{D})
        \mid
        q\in\mathbb{M}_u,\,
        \mathcal{D}\in\mathbb{D}_u
        \right\}
        \right)
        \nonumber
    \end{equation}
A higher MU indicates better retention of both neighboring and general knowledge, while an MU close to that of the ``Origin'' model suggests that model utility is well preserved after unlearning.

\subsection{Results Analysis}
\autoref{tab:tofu_main} presents the aggregated unlearning performance on TOFU under different forgetting ratios compared with baselines. 
The reported ``Forget'' and ``Retain'' are both \textit{aggregated metrics}. 
We also present the \textit{original detailed metrics} in \autoref{tab:tofu_detailed}.

\mypara{Observation} 
Across all settings, \sysname consistently achieves complete removal of target knowledge while retaining model utility almost without degradation. 
Specifically, its forgetting scores on forget01, forget05, and forget10 outperform both GA and NPO, and exceed those of the ``Retrain'' model. 
Since the ``Retrain'' model is trained without $\mathcal{D}_f$, this result indicates that, when evaluated on queries related to $\mathcal{D}_f$, \sysname behaves as if it had never been exposed to the forget data, \ie, complete removal of target $\mathcal{D}_f$.
Importantly, this improved forgetting performance is not achieved through excessive degradation of model utility. 
\sysname maintains retention scores close to those of the ``Origin'' model, demonstrating that both neighboring knowledge and general factual knowledge are well preserved.

Furthermore, \sysname remains stable as the forgetting ratio increases. 
In contrast, GA exhibits substantial and unstable utility degradation, while NPO and RMU suffer a decline in utility retention. 
Together with the results on MUSE and WMDP, these findings further present the generalizability of \sysname across copyrighted content, hazardous knowledge, and personal information.

\begin{table}[t]
\centering
\caption{
Performance on TOFU under different forgetting settings.
``Forget'' and ``Retain'' are measured by FE and MU, which are aggregated metrics.
We highlight the best performance of forgetting and retention in \textbf{bold}.
}
\label{tab:tofu_main}
\resizebox{\linewidth}{!}{
\begin{tabular}{ccccccc}
\toprule
\multirow{2}{*}{\textbf{Method}}
& \multicolumn{2}{c}{forget01}
& \multicolumn{2}{c}{forget05}
& \multicolumn{2}{c}{forget10} \\
\cmidrule(lr){2-3} \cmidrule(lr){4-5} \cmidrule(lr){6-7}
& Forget$\uparrow$ & Retain$\uparrow$
& Forget$\uparrow$ & Retain$\uparrow$
& Forget$\uparrow$ & Retain$\uparrow$ \\
\midrule
Origin
& 3.09  & 75.81
& 3.19  & 75.85
& 3.19  & 75.85 \\
Retrain
& 63.35 & 74.68
& 65.18 & 74.67
& 65.45 & 73.64 \\
\midrule
GA
& 69.46 & 66.59
& 3.89  & 29.25
& 11.17  & 50.29 \\
NPO
& 71.14 & 64.10
& 69.31 & 59.62
& 73.04 & 56.78 \\
RMU
& 73.91 & 63.24
& 70.72 & 66.42
& 73.63 & 70.91 \\
\midrule
\rowcolor[RGB]{227,234,241}
\textbf{\sysname}
& \textbf{86.76} & \textbf{75.79}
& \textbf{91.27} & \textbf{75.58}
& \textbf{88.28} & \textbf{76.05} \\
\bottomrule
\end{tabular}
}
\end{table}

\section{Detailed Analysis}
\label{sec: detailed analysis}

\subsection{Causal Effect of Various Datasets}
\label{subsec: causal effect of various datasets}
\mypara{Setup}
We visualize the Causal Effect of FFNs, using datasets with different knowledge domains.
1,000 Factual Statements consists of a collection of concise, objective assertions regarding real-world facts.
MUSE-BOOKS contains question-answering pairs derived from the Harry Potter series, while MUSE-NEWS encompasses information collected from BBC News corpora published after August 2023.
The two datasets of MUSE represent the specific knowledge domains.

The QA pairs in MUSE-BOOKS and MUSE-NEWS represent more general, context-based knowledge retrieval scenarios but do not follow a fixed atomic-fact structure. 
Consequently, their questions and answers may contain low-information generic tokens, such as articles and pronouns (\eg, ``The,'' ``A,'' or ``To''), which can interfere with the accuracy of causal tracing. 
To address this, we select perturbation tokens and answer tokens according to the \textit{information} they carry. 
For perturbation tokens, we identify semantically meaningful entities in the question and select those whose replacement changes the model’s answer. 
For answer tokens, we first identify informative entities in the answer (\eg, persons, locations, or organizations), and prefill the low-information prefix preceding the selected entity into the prompt, such that the first token of the entity serves as the target answer token. 
In this way, the first answer token used to compute the causal effect is highly informative, avoiding interference from generic tokens.

\begin{table}[]
\centering
\caption{Detailed metrics on TOFU before aggregated.}
\label{tab:tofu_detailed}
\fontsize{6}{6.5}\selectfont
\setlength{\tabcolsep}{1.6pt}
\renewcommand{\arraystretch}{1.2}

\begin{tabular}{cc|cccccc|cccccc}
\toprule
\textbf{Setting}
& \textbf{Method}
& \textbf{R}
& \textbf{P}
& \textbf{TR}
& \textbf{TE}
& \textbf{CS}
& \textbf{ES}
& \textbf{R}
& \textbf{P}
& \textbf{TR}
& \textbf{TE}
& \textbf{CS}
& \textbf{ES} \\
\midrule
\multicolumn{2}{c|}{\cellcolor[gray]{0.95}}
& \multicolumn{6}{c|}{
    \cellcolor[gray]{0.95}\textbf{Forget Set} ($\downarrow$)}
& \multicolumn{6}{c}{
    \cellcolor[gray]{0.95}\textbf{Retain Set} ($\uparrow$)} \\

\multirow[c]{3}{*}{forget01}
& Base
& 95.22 & 99.30 & 46.59 & --    & 99.05 & 92.50
& 98.10 & 98.95 & 48.28 & 96.95 & 98.84 & 96.66 \\
& Retrain
& 40.55 & 18.50 & 32.26 & --    & 76.90 & 15.00
& 98.38 & 98.88 & 47.21 & 96.96 & 99.14 & 97.66 \\
& \cellcolor[RGB]{227,234,241}\textbf{\sysname}
& \cellcolor[RGB]{227,234,241}13.11
& \cellcolor[RGB]{227,234,241}1.88
& \cellcolor[RGB]{227,234,241}23.80
& \cellcolor[RGB]{227,234,241}--
& \cellcolor[RGB]{227,234,241}22.42
& \cellcolor[RGB]{227,234,241}5.00
& \cellcolor[RGB]{227,234,241}98.23
& \cellcolor[RGB]{227,234,241}98.93
& \cellcolor[RGB]{227,234,241}48.22
& \cellcolor[RGB]{227,234,241}96.93
& \cellcolor[RGB]{227,234,241}98.86
& \cellcolor[RGB]{227,234,241}96.66 \\
\cmidrule{1-14}
\multirow[c]{3}{*}{forget05}
& Base
& 97.25 & 98.92 & 49.35 & --    & 99.00 & 94.00
& 98.10 & 98.95 & 48.28 & 96.95 & 98.84 & 96.66 \\
& Retrain
& 39.94 & 14.90 & 33.10 & --    & 73.60 & 12.50
& 97.38 & 98.97 & 47.33 & 96.94 & 98.47 & 96.00 \\
& \cellcolor[RGB]{227,234,241}\textbf{\sysname}
& \cellcolor[RGB]{227,234,241}8.59
& \cellcolor[RGB]{227,234,241}2.44
& \cellcolor[RGB]{227,234,241}20.34
& \cellcolor[RGB]{227,234,241}--
& \cellcolor[RGB]{227,234,241}10.73
& \cellcolor[RGB]{227,234,241}1.50
& \cellcolor[RGB]{227,234,241}98.59
& \cellcolor[RGB]{227,234,241}98.92
& \cellcolor[RGB]{227,234,241}48.27
& \cellcolor[RGB]{227,234,241}96.92
& \cellcolor[RGB]{227,234,241}99.00
& \cellcolor[RGB]{227,234,241}97.00 \\
\cmidrule{1-14}
\multirow[c]{3}{*}{forget10}
& Base
& 98.77 & 99.01 & 49.26 & --    & 99.58 & 97.00
& 98.10 & 98.95 & 48.28 & 96.95 & 98.84 & 96.66 \\
& Retrain
& 40.61 & 14.75 & 32.78 & --    & 72.24 & 12.33
& 97.64 & 98.96 & 46.76 & 96.95 & 98.76 & 96.33 \\
& \cellcolor[RGB]{227,234,241}\textbf{\sysname}
& \cellcolor[RGB]{227,234,241}9.00
& \cellcolor[RGB]{227,234,241}8.17
& \cellcolor[RGB]{227,234,241}20.96
& \cellcolor[RGB]{227,234,241}--
& \cellcolor[RGB]{227,234,241}14.77
& \cellcolor[RGB]{227,234,241}5.66
& \cellcolor[RGB]{227,234,241}98.12
& \cellcolor[RGB]{227,234,241}98.83
& \cellcolor[RGB]{227,234,241}48.06
& \cellcolor[RGB]{227,234,241}96.92
& \cellcolor[RGB]{227,234,241}98.82
& \cellcolor[RGB]{227,234,241}96.33 \\
\midrule
\multicolumn{2}{c|}{\cellcolor[gray]{0.95}}
& \multicolumn{6}{c|}{
    \cellcolor[gray]{0.95}\textbf{Real Authors} ($\uparrow$)}
& \multicolumn{6}{c}{
    \cellcolor[gray]{0.95}\textbf{World Facts} ($\uparrow$)} \\
\multirow[c]{3}{*}{forget01}
& Base
& 93.29 & 44.71 & 57.90 & 98.60 & 99.83 & 95.00
& 88.32 & 42.58 & 55.93 & 95.92 & 99.26 & 88.03 \\
& Retrain
& 91.79 & 45.57 & 60.64 & 98.48 & 96.55 & 85.00
& 89.31 & 42.42 & 55.90 & 95.73 & 95.76 & 72.64 \\
& \cellcolor[RGB]{227,234,241}\textbf{\sysname}
& \cellcolor[RGB]{227,234,241}93.29
& \cellcolor[RGB]{227,234,241}45.32
& \cellcolor[RGB]{227,234,241}58.90
& \cellcolor[RGB]{227,234,241}98.60
& \cellcolor[RGB]{227,234,241}99.79
& \cellcolor[RGB]{227,234,241}95.00
& \cellcolor[RGB]{227,234,241}87.46
& \cellcolor[RGB]{227,234,241}42.66
& \cellcolor[RGB]{227,234,241}55.94
& \cellcolor[RGB]{227,234,241}95.88
& \cellcolor[RGB]{227,234,241}98.73
& \cellcolor[RGB]{227,234,241}83.76 \\
\cmidrule{1-14}
\multirow[c]{3}{*}{forget05}
& Base
& 93.29 & 44.71 & 57.90 & 98.60 & 99.83 & 95.00
& 88.31 & 42.58 & 55.93 & 95.92 & 99.26 & 88.03 \\
& Retrain
& 92.65 & 46.03 & 60.51 & 98.45 & 96.65 & 87.00
& 88.10 & 42.38 & 56.16 & 95.59 & 95.37 & 71.79 \\
& \cellcolor[RGB]{227,234,241}\textbf{\sysname}
& \cellcolor[RGB]{227,234,241}94.80
& \cellcolor[RGB]{227,234,241}45.17
& \cellcolor[RGB]{227,234,241}58.31
& \cellcolor[RGB]{227,234,241}98.62
& \cellcolor[RGB]{227,234,241}99.04
& \cellcolor[RGB]{227,234,241}93.00
& \cellcolor[RGB]{227,234,241}87.03
& \cellcolor[RGB]{227,234,241}42.83
& \cellcolor[RGB]{227,234,241}56.24
& \cellcolor[RGB]{227,234,241}96.12
& \cellcolor[RGB]{227,234,241}98.23
& \cellcolor[RGB]{227,234,241}80.34 \\
\cmidrule{1-14}
\multirow[c]{3}{*}{forget10}
& Base
& 93.29 & 44.71 & 57.90 & 98.60 & 99.83 & 95.00
& 88.32 & 42.58 & 55.93 & 95.92 & 99.26 & 88.03 \\
& Retrain
& 91.55 & 44.58 & 58.18 & 98.54 & 96.67 & 86.00
& 90.17 & 41.11 & 53.35 & 95.75 & 94.97 & 70.94 \\
& \cellcolor[RGB]{227,234,241}\textbf{\sysname}
& \cellcolor[RGB]{227,234,241}94.30
& \cellcolor[RGB]{227,234,241}46.44
& \cellcolor[RGB]{227,234,241}59.95
& \cellcolor[RGB]{227,234,241}98.64
& \cellcolor[RGB]{227,234,241}98.28
& \cellcolor[RGB]{227,234,241}93.00
& \cellcolor[RGB]{227,234,241}87.89
& \cellcolor[RGB]{227,234,241}43.56
& \cellcolor[RGB]{227,234,241}57.33
& \cellcolor[RGB]{227,234,241}96.14
& \cellcolor[RGB]{227,234,241}97.54
& \cellcolor[RGB]{227,234,241}80.34 \\
\bottomrule
\end{tabular}
\end{table}

\mypara{Observations}
\autoref{fig: causal trace compare} presents the visualization results.
Despite the differences in semantic content across these datasets, the distributions of $\text{CE}_{FFN}$ exhibit a remarkable degree of consistency. 
The FFNs with significant causal effect are predominantly locate within the early layers of the model. 
This distributional consistency demonstrates that diverse categories of factual knowledge follow similar storage patterns within the internal architecture of LLMs.
Consequently, the knowledge storage patterns derived from the 1,000 Factual Statements, which encompass a broad range of real-world factual knowledge, present cross-domain universality and can be robustly generalized to other knowledge domains. 
This empirical result supports the adoption of 1,000 Factual Statements for layer selection across diverse unlearning tasks. 
By doing so, we effectively circumvent the prohibitive computational overhead associated with per-task causal tracing, significantly enhancing the scalability and practicality of our unlearning framework.
This empirical result validates our design of reusing the FFNs, which are identified via the 1,000 Factual Statements, for diverse unlearning tasks, thereby circumventing the prohibitive computational overhead caused by implementing causal tracing for each unlearning dataset.

\begin{figure}[!t]
    \centering
    \includegraphics[width=0.6\hsize]{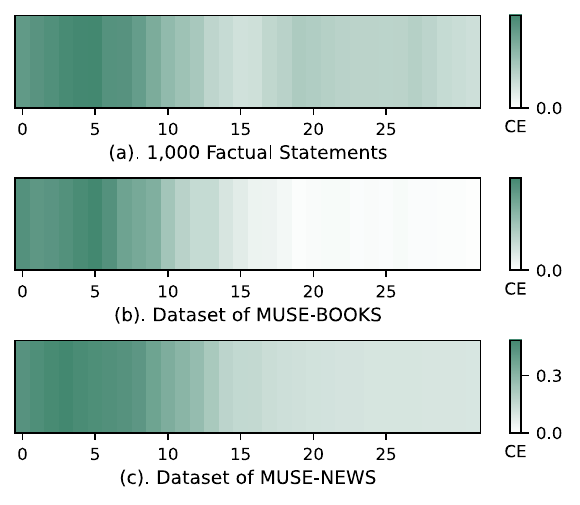}
    \caption{
    Visualization of Causal Effect of FFN across different datasets.
    The darker colors indicate stronger results.
    }
    \label{fig: causal trace compare}
\end{figure}

\subsection{Unlearning Layers Analysis}
\label{subsec: unlearning layers analysis}

\mypara{Setup}
We evaluate unlearning performance across various unlearning layers to investigate how the selection of target layers influences the overall results.
\autoref{fig:CE_zephyr} presents the $\text{CE}_{FFN}$ and $\text{CE}_{Attn}$ scores, and \autoref{fig:layer selection} shows the unlearning performance with different unlearning layers.

\mypara{Advantage of FFN Dominance}
\autoref{fig:layer selection} demonstrates that unlearning performance is highly sensitive to the functional properties of the targeted layers. 
Unlearning layers with high $\text{CE}_{FFN}$, as identified in \autoref{fig:CE_zephyr}, significantly outperform layers with high $\text{CE}_{Attn}$ in balancing forgetting and utility.
This indicates that unlearning on layers responsible for knowledge storage enables precise removal of target knowledge.

\mypara{Comparison with Naive Selection}
\autoref{fig:layer selection} also reveals the limitations of a naive strategy that directly selects layers with the highest $\text{CE}_{FFN}$ scores. 
In our experiments, these peak values coincidentally occur in the shallow layers ($l_3 \sim l_6$), as presented in \autoref{fig:CE_zephyr}. 
However, updating these layers yields poor unlearning performance, failing to remove knowledge while severely compromising the model utility. 
This may stem from that the first several layers handle generic, low-level features shared across tasks~\cite{jawahar2019bertlearn}.
Updating these lacks the semantic specificity required for targeted knowledge removal, thus rendering them unsuitable as unlearning layers.

\mypara{The Limits of Depth}
We have discussed the limitations of shallow layers with high $\text{CE}_{FFN}$ scores.
However, deeper layers do not indicate better unlearning performance. 
While deeper layers theoretically aggregate more knowledge due to residual connections, our findings reveal a critical boundary.
There exist the transition layers ($l_{10} \sim l_{15}$) between the FFN-dominant and the MHSA-dominant zone.
These layers might also exhibit high $CE_{FFN}$, especially those located at the terminus of the FFN-dominant zone, for example, $l_{10} \sim l_{12}$.
However, updating these layers could cause severe utility degradation. 
This potentially arises from the disruption of the functional transition between knowledge storage and semantics aggregation.

\mypara{Consecutive Layers}
The right plot of \autoref{fig:layer selection} also reveals a slight advantage of consecutive layers.
While some discontinuous selections may achieve slightly higher forgetting, they often suffer from a decline in model utility. 
A selection of continuous layers ensures a unified modification of the targeted knowledge representation, whereas skipping intermediate layers may disrupt the structural coherence of the representation space, leading to functional instability.
Therefore, consecutive unlearning layers could simplify the complexity of the situation where layers with high $\text{CE}_{FFN}$ are discrete.

\begin{figure}[!t]
    \centering
    \includegraphics[width=1.0\hsize]{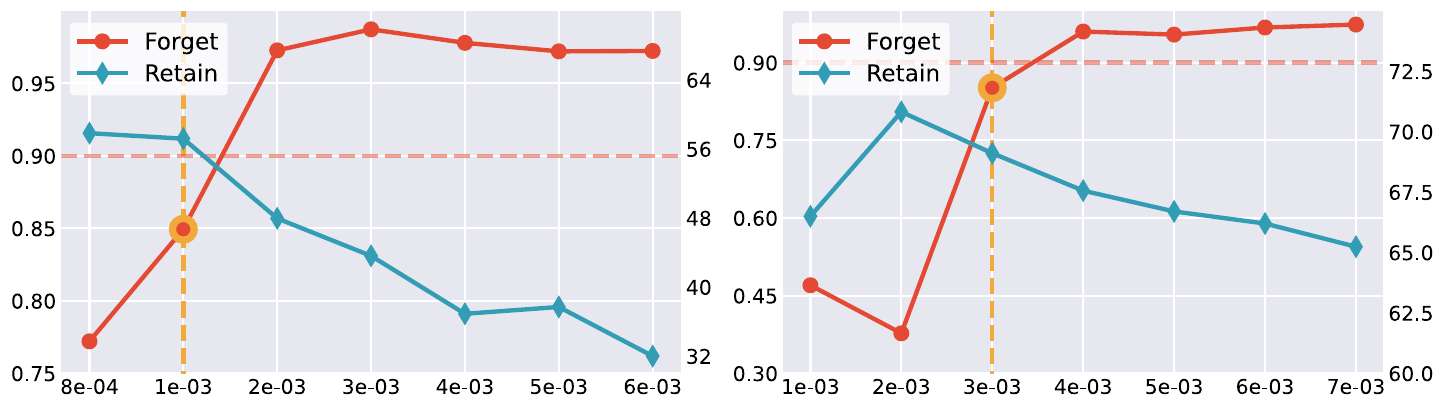}
    \caption{
    Unlearning with fixed $\beta=0.1$ and different $\tau$.
    The left plot employs Zephyr-7B-Beta, while the right plot corresponds to Qwen-3-8B.
    The highlighted value of $\tau$ is the stability boundary. 
    The left y-axis represents KRD, while the right y-axis denotes the MMLU score. 
    A reference line at $\text{KRD} = 0.9$ is plotted. 
    The ``stability boundary'' is defined as the specific $\tau$ where KRD begins to fall lower than the reference.
    }
    \label{fig: tau boundary}
\end{figure}

\subsection{Stability Boundary}
\label{subsec: stability boundary}
\mypara{Setup}
We conduct an initial sensitivity analysis on the hyperparameter $\tau$. 
In this setup, we fix $\beta=0.1$, representing the most aggressive forgetting in our experiments.
Under this fixed constraint, we progressively decrease the value of $\tau$.
Particularly, we define the ``stability boundary'' as the specific value of $\tau$ at which KRD first falls below 0.9.

\mypara{Observation}
\autoref{fig: tau boundary} illustrates the impact of different values of $\tau$ on unlearning performance with a fixed $\beta = 0.1$. 
We observe a significant transition in forgetting at a specific $\tau$, which we identify as the \textit{stability boundary}.
At or below this point, \sysname fails to remove the target knowledge thoroughly. 
This phenomenon arises from the entanglement of knowledge representations, which causes the relaxed null-space projection to exert a global preservation effect across multiple domains.
Beyond merely preserving the retaining knowledge, it excessively constrains perturbations across the entire representation space. 
Consequently, the deviations introduced for knowledge removal are also eliminated.
In practice, we adopt a value of $\tau$ slightly larger than the stability boundary, such as $2 \times 10^{-3}$ for Zephyr-7B-Beta and $4 \times 10^{-3}$ for Qwen-3, and then gradually increase $\beta$ across $\{ 0.1, 1.0, 2.0\}$ to achieve the unlearning balance.

\subsection{FFN-Centered Update Design}
\label{subsec:ffn_mhsa}
\begin{table}[t]
\centering
\caption{Ablation of the FFN-centered update design.
}
\label{tab:ffn_mhsa}
\scriptsize
\renewcommand{\arraystretch}{1.0}
\setlength{\tabcolsep}{1.0pt}

\begin{minipage}[t]{\linewidth}
\vspace{0pt}
\centering
\textbf{(a) MUSE}\\[1pt]

\begin{tabular*}{\linewidth}{
@{\extracolsep{\fill}}
lccccc
@{}
}
\toprule
\multirow{2}{*}{}
& \multicolumn{4}{c}{\textbf{Forgetting Quality}}
& \multicolumn{1}{c}{\textbf{Retaining Utility}} \\
\cmidrule(lr){2-5}
\cmidrule(l){6-6}
&
\makecell{VerbMem\\$\mathcal{D}_f\downarrow$}
&
\makecell{KnowMem\\$\mathcal{D}_f\downarrow$}
&
\makecell{PrivLeak\\$\rightarrow 0$}
&
\makecell{KRD\\$\rightarrow 1$}
&
\makecell{KnowMem\\$\mathcal{D}_r\uparrow$} \\
\midrule

\multicolumn{6}{c}{\textbf{BOOKS}} \\
\cmidrule(){1-6}
Origin
& 99.59 & 58.76 & 56.68 & 0.0000 & 67.00 \\
Retrain
& 14.35 & 28.90 & 0.00 & 1.0000 & 74.38 \\
FFN (KUDA)
& 0.05 & 26.73 & 16.55 & 0.8792 & 61.97 \\
FFN+MHSA
& 15.67 & 28.32 & 57.26 & 0.0003 & 51.61 \\
\midrule

\multicolumn{6}{c}{\textbf{NEWS}} \\
\cmidrule(){1-6}
Origin
& 58.42 & 63.41 & 99.84 & 0.0000 & 54.96 \\
Retrain
& 20.80 & 33.30 & 0.00 & 1.0000 & 54.68 \\
FFN (KUDA)
& 8.71 & 0.26 & 20.09 & 0.9226 & 53.08 \\
FFN+MHSA
& 21.27 & 41.05 & 97.56 & 0.0653 & 34.19 \\
\bottomrule
\end{tabular*}
\end{minipage}

\par\vspace{3pt}

\begin{minipage}[t]{\linewidth}
\vspace{0pt}
\centering
\textbf{(b) WMDP}\\[1pt]

\begin{tabular*}{\linewidth}{
@{\extracolsep{\fill}}
lccccc
@{}
}
\toprule
\multirow{2}{*}{Update Scope}
& \multicolumn{4}{c}{\textbf{Forgetting Quality}}
& \multicolumn{1}{c}{\textbf{Retaining Utility}} \\
\cmidrule(lr){2-5}
\cmidrule(l){6-6}
&
\makecell{Acc-Bio\\$\downarrow$}
&
\makecell{Acc-Cyber\\$\downarrow$}
&
\makecell{Acc-Avg\\$\downarrow$}
&
\makecell{KRD\\$\rightarrow 1$}
&
\makecell{MMLU\\$\uparrow$} \\
\midrule

\multicolumn{6}{c}{\textbf{Zephyr-7B-Beta}} \\
\cmidrule(){1-6}
Origin
& 63.82 & 43.78 & 51.59 & 0.0000 & 58.13 \\
FFN (KUDA)
& 29.92 & 26.82 & 28.03 & 0.8879 & 57.14 \\
FFN+MHSA
& 41.40 & 27.93 & 33.20 & 0.6858 & 33.14 \\
\midrule

\multicolumn{6}{c}{\textbf{Llama-3.1-8B}} \\
\cmidrule(){1-6}
Origin
& 59.62 & 41.84 & 48.80 & 0.0000 & 63.31 \\
FFN (KUDA)
& 29.22 & 25.00 & 26.28 & 0.9351 & 60.64 \\
FFN+MHSA
& 28.33 & 25.00 & 25.92 & 0.9495 & 52.03 \\
\midrule

\multicolumn{6}{c}{\textbf{Qwen3-8B}} \\
\cmidrule(){1-6}
Origin
& 74.16 & 45.14 & 56.47 & 0.0000 & 73.01 \\
FFN (KUDA)
& 28.43 & 26.17 & 27.05 & 0.9360 & 68.92 \\
FFN+MHSA
& 27.47 & 25.86 & 26.53 & 0.9535 & 45.76 \\
\bottomrule
\end{tabular*}
\end{minipage}
\end{table}
To further validate the FFN-centered design of \sysname, we extend the update scope to MHSA and examine whether it yields a better forgetting-retention trade-off. 

\mypara{Setup}
We compare the original \sysname configuration, which updates $W_{out}$ in the FFNs of the selected unlearning layers, with an extended variant that additionally updates $W_V$ in their MHSA modules. 
The knowledge-related contribution of MHSA arises from its ability to extract and aggregate information from the context. 
Specifically, $W_Q$ and $W_K$ determine how attention is allocated over the contextual tokens, whereas $W_V$ determines the information extracted from and aggregated across these tokens. 
We therefore update $W_V$ to intervene in the knowledge-related information aggregated by MHSA, while keeping the mechanism of attention allocation intact. 
Except for the expanded update parameters, all other components and experimental settings remain consistent.

\mypara{Observation}
As presented in ~\autoref{tab:ffn_mhsa}, extending the update scope to MHSA does not yield the expected improvement in knowledge removal, but instead substantially degrades retaining utility.
Specifically, on MUSE, jointly updating FFNs and MHSA degrades both forgetting quality and retaining utility on both BOOKS and NEWS.
A possible explanation is that the knowledge-related contribution of MHSA primarily arises from aggregating the question context, rather than from directly encoding the target knowledge associations. 
Once $W_V$ is updated, the representation-deviation objective may exploit a shortcut by altering the aggregated contextual information, thereby reducing the intervention to the target associations encoded in FFNs. 
Consequently, the target knowledge may remain accessible under specific contexts, while the resulting behavioral changes fail to eliminate the statistical traces from the training data, as supported by the PrivLeak scores remaining close to those of the ``Origin'' model.
Moreover, because contextual aggregation is broadly involved in general language modeling, modifying MHSA  substantially degrades model utility.
On WMDP, both forgetting quality and retaining utility decrease substantially for Zephyr-7B-Beta. 
For Llama-3.1-8B and Qwen3-8B, the slight forgetting improvements are accompanied by substantial utility degradation, suggesting that these marginal gains may arise primarily from degraded contextual understanding rather than targeted knowledge removal.
This further indicates that modifying MHSA influences contextual processing without yielding unlearning benefits.

\mypara{Disccusion}
Although MHSA may also contribute to knowledge associations, such contributions arise from the contextual mechanisms that are fundamental to general language modeling~\cite{geva2021transformer,meng2022locating}. 
Consequently, updating MHSA risks causing broader degradation in general utility. 
Our experiments with jointly updating FFN and MHSA show that expanding the updates to MHSA yields no consistent forgetting gains, while degrading the utility. 
This supports our FFN-centered design.

\subsection{Efficiency Comparison}
\label{subsec: efficiency comparison}
To further demonstrate that \sysname incurs limited additional computational overhead, we compare its computational cost with representative baselines under the same settings.

\mypara{Setup}
We measure the runtime cost of GA, NPO, SimNPO, RMU, and \sysname under the same hardware and training configuration, using two NVIDIA L40 GPUs, the same batch size of 4, and 10 unlearning epochs. 
For each setting, KUDA and RMU update the same four layers.

\mypara{Computational Cost Analysis}
We consider two categories of representative baselines: full parameter and partial parameter update methods. 
For full parameter update methods, we consider GA, SimNPO, and NPO. 
All three methods compute their losses based on the model outputs and then update all model parameters.
Their computational costs are therefore dominated by full-model forward and backward propagation.
GA and SimNPO depend only on the current model and thus have similar workflows, resulting in comparable costs. 
In contrast, NPO additionally requires a frozen reference model to perform an extra forward pass on $\mathcal{D}_f$, leading to higher computational overhead. 
Overall, their runtime costs follow the trend $\mathrm{GA}\approx\mathrm{SimNPO}<\mathrm{NPO}$.

For partial parameter update methods, RMU and \sysname update the same four layers.
Accordingly, their computational difference mainly arises from method-specific components.
The additional computations of \sysname mainly include causal tracing, input feature capture, SVD, and gradient projection. 
Causal tracing constitutes the primary one-time cost and is used to select the layers to be updated. 
It only needs to be performed once for \textit{each model}, and the resulting layer selection can be reused across different datasets. 
Input feature capture and SVD are performed only once for \textit{each dataset}, and the results are reused throughout the subsequent gradient updates. 
Therefore, during iterative optimization, the main recurring overhead comes from gradient projection at each update step.

\begin{table}[t]
    \centering
    \caption{
    Total runtime for 10 unlearning epochs across different datasets and models.
    ``\sysname w/ Causal'' includes the one-time causal-tracing cost,
    whereas ``\sysname w/o Causal'' reuses the identified layers and excludes this cost.
    Runtime is reported in HH:MM:SS.
    }
    \label{tab:efficiency_full}
    \setlength{\tabcolsep}{4pt}
    \renewcommand{\arraystretch}{1.0}
    \begin{tabular}{lccccc}
        \toprule
        \multirow{2}{*}{\textbf{Method}}
        & \multicolumn{3}{c}{\textbf{WMDP}}
        & \multicolumn{2}{c}{\textbf{MUSE}} \\
        \cmidrule(lr){2-4}
        \cmidrule(lr){5-6}
        & \makecell{\textbf{Zephyr-Beta}}
        & \makecell{\textbf{Llama-3.1}}
        & \makecell{\textbf{Qwen-3}}
        & \makecell{\textbf{BOOKS}}
        & \makecell{\textbf{NEWS}} \\
        \midrule
        GA
        & 04:08:05
        & 04:31:58
        & 04:47:23
        & 05:32:00
        & 03:55:39 \\

        SimNPO
        & 04:06:17
        & 04:35:24
        & 04:48:46
        & 05:35:19
        & 03:51:08\\

        NPO
        & 04:30:36
        & 04:57:22
        & 05:26:34
        & 06:45:02
        & 04:27:25\\

        RMU
        & 00:54:43
        & 01:16:14
        & 01:15:26
        & 01:53:07
        & 01:24:15 \\

        \midrule

        \makecell[l]{\sysname w/ Causal}
        & 03:22:20
        & 04:34:29
        & 04:25:57
        & 04:29:51
        & 03:28:41 \\

        \makecell[l]{\sysname w/o Causal}
        & 01:11:57
        & 01:37:29
        & 01:44:12
        & 02:08:45
        & 01:38:17 \\
        \bottomrule
    \end{tabular}
\end{table}

\mypara{Observation}
Overall, as shown in ~\autoref{tab:efficiency_full}, \sysname achieves a favorable trade-off between efficiency and effectiveness. 
Concretely, when the layer-selection results are reused (\ie, w/o causal tracing), \sysname requires substantially less computational cost than GA, the lower-cost option in full parameter baselines, across all experimental settings. 
Even when including the one-time cost of causal tracing, \sysname's total runtime remains comparable to or lower than that of GA, demonstrating that the additional computations do not offset its computational advantage over full parameter update methods. 
Compared with RMU, when reusing the causal-tracing results, \sysname incurs only limited additional overhead in exchange for better unlearning performance, while the two methods remain within a similar runtime scale.

\section{Adversarial Attacks}
\label{sec: Adversarial Attacks}

\subsection{Min-K$\%$ Prob MIA}
One of the ideal goals in LLM unlearning is to make the forget set $\mathcal{D}_f$ statistically indistinguishable from unseen data $\mathcal{D}_{holdout}$. 
MIA directly quantifies the indistinguishability by exploiting distributional differences in certain statistics (\eg, loss) between member and non-member data.
To evaluate whether the unlearned model retains residual membership information, we employ Min-K$\%$ Prob MIA~\cite{shi2024detecting}, which is designed for language models based on the loss.
It is grounded in the hypothesis that unseen (non-member) examples are more likely to contain some outlier tokens with extremely low conditional probabilities, whereas seen (member) examples tend to exhibit more uniform token-level likelihoods. 

\mypara{Implement and Metrics}
Formally, given an input sequence $x$, the method computes the conditional log-probability for each token, selects the K$\%$ tokens with the minimum probabilities to construct the Min-K$\%$ set for $x$.
Then, it calculates their average log-likelihood as the discrimination score. 

Certain methods of LLM unlearning, such as GA, achieve forgetting by intentionally inflating the loss, potentially leading to abnormally high discrimination scores. 
Therefore, we transform the discrimination score into an AUC score. 
Subsequently, following the definition of \textit{PrivLeak} in Appendix~\ref{subsec: Metrics}, we normalize AUC score by comparing it with that of the ``retrain'' model to generate \textit{Relative AUC}, which is the metric in \autoref{fig:Attack}a.
As its value approaches zero, the model performs similarly with the ``retrain'' baseline regarding membership discriminability.
This indicates that the adversaries gain no more information than the model retrained from scratch.

\subsection{FOCUSONKEY}
FOCUSONKEY~\cite{zhang2025focusonkey} is a prompt-based knowledge recovery attack designed to verify whether the forgotten knowledge is truly removed from model parameters. 
The core insight is that since existing unlearning methods primarily work by disrupting the model's attention to keywords in the prompt, explicitly emphasizing these keywords can recover supposedly forgotten information~\cite{zhang2025focusonkey}. 
Specifically, FOCUSONKEY first identifies the set of key tokens $K_{pre}$ that the model, before unlearning, focuses on when correctly answering questions. 
It then appends a minimal emphasis phrase to the original prompt (\eg, ``Focus on $[K_s]$ to answer'', where $K_s  \subseteq K_{pre}$) to highlight the relevant key tokens. 
By exploiting such prompts, FOCUSONKEY can recover knowledge under a black-box setting without access to model parameters, relying solely on greedy decoding. 

\mypara{Implement and Metrics}
We implement FOCUSONKEY following Zhang \etal~\cite{zhang2025focusonkey}. 
We employ the Prometheus-7B-v2.0~\cite{kim2024Prometheus} to construct an LLM-as-a-judge framework to automatically assess whether the target model correctly answers relevant questions, thereby computing the answer accuracy (Accuracy). 
In \autoref{fig:Attack}b, we adopt \textit{Accuracy} as the metric to quantify the degree of knowledge recovery.

\begin{table}[t]
    \centering
    \caption{
    Robustness of KUDA against quantization-based knowledge recovery on MUSE and WMDP.}
    \label{tab:quantization}
    \scriptsize
    \renewcommand{\arraystretch}{1.0}

    \textbf{(a) MUSE}\\[2pt]
    \setlength{\tabcolsep}{5.0pt}
    \begin{tabular}{lcccc}
        \toprule
        \multirow{2}{*}{\textbf{Unlearning Method}}
        & \multicolumn{3}{c}{\textbf{Forgetting Quality}}
        & \multicolumn{1}{c}{\textbf{Retaining Utility}} \\
        \cmidrule(lr){2-4}
        \cmidrule(lr){5-5}
        & \shortstack{VerbMem\\$\mathcal{D}_f \downarrow$}
        & \shortstack{KnowMem\\$\mathcal{D}_f \downarrow$}
        & \shortstack{PrivLeak\\$\rightarrow 0$}
        & \shortstack{KnowMem\\$\mathcal{D}_r \uparrow$} \\
        \midrule

        \multicolumn{5}{c}{\textbf{MUSE-BOOKS}} \\
        \cmidrule(){1-5}
        Origin       & 99.59 & 58.76 & 56.68 & 67.00 \\
        Retrain      & 14.35 & 28.90 &  0.00 & 74.38 \\
        \cmidrule(l){2-5}
        KUDA (BF16)  &  0.05 & 26.73 & 16.55 & 61.97 \\
        KUDA (8-bit) &  0.08 & 29.34 & 15.37 & 64.59 \\
        KUDA (4-bit) &  1.56 & 16.73 & 28.71 & 47.96 \\

        \midrule
        \multicolumn{5}{c}{\textbf{MUSE-NEWS}} \\
        \cmidrule(){1-5}
        Origin       & 58.42 & 63.41 & 99.84 & 54.96 \\
        Retrain      & 20.80 & 33.30 &  0.00 & 54.68 \\
        \cmidrule(l){2-5}
        KUDA (BF16)  &  8.71 &  0.26 & 20.09 & 53.08 \\
        KUDA (8-bit) &  9.28 &  0.55 & 14.82 & 53.22 \\
        KUDA (4-bit) &  7.90 &  0.39 & 18.76 & 49.13 \\
        \bottomrule
    \end{tabular}

    \vspace{2pt}

    \textbf{(b) WMDP}\\[2pt]
    \setlength{\tabcolsep}{4.0pt}
    \begin{tabular}{lcccc}
        \toprule
        \multirow{2}{*}{\textbf{Unlearning Method}}
        & \multicolumn{3}{c}{\textbf{Forgetting Quality}}
        & \multicolumn{1}{c}{\textbf{Retaining Utility}} \\
        \cmidrule(lr){2-4}
        \cmidrule(lr){5-5}
        & Acc-Bio$\downarrow$
        & Acc-Cyber$\downarrow$
        & Acc-Avg$\downarrow$
        & MMLU$\uparrow$ \\
        \midrule

        \multicolumn{5}{c}{\textbf{Zephyr-7B-Beta}} \\
        \cmidrule(){1-5}
        Origin       & 63.82 & 43.78 & 51.59 & 58.13 \\
        \cmidrule(l){2-5}
        KUDA (BF16)  & 29.92 & 26.82 & 28.03 & 57.14 \\
        KUDA (8-bit) & 29.93 & 26.32 & 27.73 & 57.06 \\
        KUDA (4-bit) & 29.30 & 26.37 & 27.51 & 54.83 \\

        \midrule
        \multicolumn{5}{c}{\textbf{Llama-3.1-8B}} \\
        \cmidrule(){1-5}
        Origin       & 59.62 & 41.84 & 48.80 & 63.31 \\
        \cmidrule(l){2-5}
        KUDA (BF16)  & 29.22 & 25.13 & 26.28 & 60.64 \\
        KUDA (8-bit) & 29.37 & 25.40 & 26.59 & 59.56 \\
        KUDA (4-bit) & 25.53 & 25.00 & 25.16 & 52.94 \\

        \midrule
        \multicolumn{5}{c}{\textbf{Qwen3-8B}} \\
        \cmidrule(){1-5}
        Origin       & 74.16 & 45.14 & 56.47 & 73.01 \\
        \cmidrule(l){2-5}
        KUDA (BF16)  & 28.43 & 26.17 & 27.05 & 68.92 \\
        KUDA (8-bit) & 27.57 & 26.87 & 27.14 & 69.13 \\
        KUDA (4-bit) & 26.70 & 25.00 & 25.60 & 65.12 \\
        \bottomrule
    \end{tabular}
\end{table}

\subsection{Quantization-based Knowledge Recovery}
\label{subsec: quantization}
Zhang \etal~\cite{zhang2025quantization} show that post-unlearning quantization on the models may eliminate the parameter changes introduced by unlearning and consequently recover the supposedly removed knowledge. 
We evaluate whether the unlearning results of \sysname remains effective after quantization.

\mypara{Implement and Metrics}
We perform KUDA-unlearning in BF16 and quantize the unlearned models to 4-bit and 8-bit precision following the protocol of Zhang \etal~\cite{zhang2025quantization}. 
No additional training or model updates are performed after quantization. 
We evaluate the quantized models using the same forgetting and retention metrics as in the main experiments.

\mypara{Observation}
The results in \autoref{tab:quantization} indicate that \sysname effectively resists knowledge recovery under both 8-bit and 4-bit quantization across different models and datasets. 
On MUSE, the forgetting metrics of the quantized models remain generally close to those of their BF16 counterparts. 
Although 4-bit quantization introduces fluctuations in several individual metrics, it does not lead to consistent knowledge recovery across metrics. 
On WMDP, the ability to utilize hazardous knowledge of the quantized models remains comparable to that of the BF16 models, indicating that \sysname preserves its forgetting effectiveness. 
Notably, 4-bit quantization causes a certain degree of utility degradation. 
This observation aligns with that of Zhang \etal~\cite{zhang2025quantization} and may reflect the general performance degradation incurred by low-bit quantization, which is not accompanied by knowledge recovery.

\end{document}

%% file: mydefs.tex
\usepackage{algorithm}
\usepackage{algpseudocode}
\usepackage{url}
\usepackage{amsmath}
\usepackage{amssymb}
\usepackage{xspace}
\usepackage{multirow,multicol}
\usepackage[table]{xcolor}
\usepackage{mathtools}
\usepackage{color}
\let\labelindent\relax
\usepackage{enumitem}
\usepackage{hyperref} 
\usepackage{dsfont}
\usepackage{amsthm}

\usepackage{microtype}

\usepackage{tabularx}
\usepackage{booktabs}
\usepackage[normalem]{ulem}
\useunder{\uline}{\ul}{}
\usepackage{makecell}

\usepackage{caption,subcaption}
\setlist[itemize]{leftmargin=*}

\usepackage{xcolor}
\hypersetup{
  colorlinks,
  linkcolor={blue!70!green},
  citecolor={green!70!blue},
  urlcolor={orange!70!red}
}

\newcommand{\etal}{\textit{et al.}\xspace}
\newcommand{\ie}{\textit{i.e.}\xspace}
\newcommand{\eg}{\textit{e.g.}\xspace}

\newcommand{\mypara}[1]{\smallskip\noindent\textbf{#1.}\xspace}

\newtheorem{definition}{Definition}
\newtheorem{theorem}{Theorem}

\usepackage{etoolbox}

\newcommand{\sysname}{\textsc{{KUDA}}\xspace}

\usepackage{etoolbox}
\makeatletter

\newtoggle{inappendix}
\togglefalse{inappendix}

\apptocmd{\appendix}{\toggletrue{inappendix}}{}{\errmessage{failed to patch}}